\RequirePackage{xcolor}
\documentclass[journal,twoside,web]{ieeecolor}
\usepackage{generic}
\usepackage[noadjust]{cite}
\let\labelindent\relax
\usepackage{amssymb,bm}
\usepackage{mathtools}
\usepackage{tikz}
\usepackage{tkz-graph}
\usetikzlibrary{arrows}
\usepackage{pgfplots}
\usepackage{hyperref}

\newcommand{\Scal}{\mathcal{S}}
\newcommand{\Xcal}{\mathcal{X}} 
\newcommand{\Ncal}{\mathcal{N}}

\usetikzlibrary{arrows, positioning, shapes, calc, fit, patterns}

\usepackage{amsfonts}
\usepackage{dsfont}
\usepackage{amsmath}
\usepackage{amssymb}
\usepackage{bm}
\usepackage{subcaption}
\usepackage{comment}

\usepackage{enumitem}

\setenumerate{label=\roman*), ref=(\roman*)}
\tikzset{
    node style/.style={circle, draw, fill=white, inner sep=1pt, minimum size=1.5em},
    main line/.style={thick}
}

\usepackage{xcolor}
\definecolor{myblue}{HTML}{006AA3}
\definecolor{mygreen}{HTML}{5B9A64}
\definecolor{myred}{HTML}{C1121F}
\definecolor{myorange}{HTML}{FAA00F}
\definecolor{mypurple}{HTML}{6D4C94}

\usepackage{tikz}
\usepackage{pgfplots}
\pgfplotsset{compat=1.18}
\usepgflibrary{arrows.meta} 

\usetikzlibrary{arrows.meta,positioning,backgrounds,shapes,calc,fit,hobby,decorations.markings,3d,calc}
\pgfdeclarelayer{background}
\pgfsetlayers{background,main}
\usepgfplotslibrary{fillbetween}

\newcommand{\mc}{\mathcal}

\newcommand{\Vcal}{\mc{V}}
\newcommand{\Gcal}{\mc{G}}
\newcommand{\Ecal}{\mc{E}}

\newcommand{\norm}[1]{\left\| #1 \right\|}
\newcommand{\abs}[1]{\left| #1 \right|}

\newcommand{\N}{\mathbb{N}}

\newcommand{\R}{\mathbb{R}}

\newcommand{\one}{\mathbf{1}}

\newcommand{\zero}{\mathbf{0}}

\usepackage{amsthm}

\newtheoremstyle{ieeeconf}
  {0pt}   
  {0pt}   
  {\normalfont}  
  {\parindent}       
  {\itshape} 
  {:}         
  { } 
  {\thmname{#1} \thmnumber{#2}\thmnote{ (#3)}} 
\makeatletter
\renewenvironment{proof}[1][\proofname]{\par
  \pushQED{\qed}%
  \normalfont \topsep\z@
  \trivlist
  \item[\hskip2em
        \itshape
    #1\@addpunct{:}]\ignorespaces
}{%
  \popQED\endtrivlist\@endpefalse
}
\makeatletter

\theoremstyle{ieeeconf}

\newtheorem{theorem}{Theorem}
\newtheorem{definition}{Definition}
\newtheorem{proposition}{Proposition}

\newtheorem{lemma}{Lemma}
\newtheorem{corollary}{Corollary}

\newtheorem{standassumption}{Standing Assumption}
\newtheorem{remark}{Remark}

\newtheorem{example}{Example}

\begin{document}
\bstctlcite{IEEEexample:BSTcontrol}

\title{From Consensus to Robust Clustering: Multi-Agent Systems with Nonlinear Interactions}

\author{Anthony Couthures$^1$
, Gustave Bainier$^{2}$
, Vineeth Satheeskumar Varma$^{1,3}$
, Samson Lasaulce$^{1}$
, Irinel-Constantin Mor\u{a}rescu$^{1,3}$
\thanks{*This work was supported by the CNRS MITI project BLESS and by the project DECIDE funded under the PNRR I8 scheme by the Romanian Ministry of Research.}
\thanks{$^1$Universit\'e de Lorraine, CNRS, CRAN, F-54000 Nancy, France. {\tt\small anthony.couthures@univ-lorraine.fr}}%
\thanks{$^2$ Dept. of Electrical Engineering and Computer Science, University of Liège, 4000 Liège, Belgium.}
\thanks{$^3$Automation Department, Technical University of Cluj-Napoca, Memorandumului 28, 400114 Cluj-Napoca, Romania.}%
}



\maketitle

\begin{abstract}
    This paper establishes a theoretical framework to describe the transition from consensus to stable clustering in multi-agent systems with nonlinear, cooperative interactions. We first establish a sharp threshold for consensus. For a broad class of non-decreasing, Lipschitz-continuous interactions, an explicit inequality linking the interaction's Lipschitz constant to the second-largest eigenvalue of the normalized adjacency matrix of the interaction graph confines all system equilibria to the synchronization manifold. This condition is shown to be a sharp threshold, as its violation permits the emergence of non-synchronized equilibria. We also demonstrate that such clustered states can only arise if the interaction law itself possesses specific structural properties, such as unstable fixed points. For the clustered states that emerge, we introduce a formal framework using Input-to-State Stability (ISS) theory to quantify their robustness. This approach allows us to prove that the internal cohesion of a cluster is robust to perturbations from the rest of the network. The analysis reveals a fundamental principle: cluster coherence is limited not by the magnitude of external influence, but by its heterogeneity across internal nodes. This unified framework, explaining both the sharp breakdown of consensus and the quantifiable robustness of the resulting modular structures, is validated on Zachary’s Karate Club network, used as a classic benchmark for community structure.
\end{abstract}


\section{Introduction} \label{sec:introduction}

A central theme in the study of multi-agent systems is the emergence of collective behavior from local interactions. One of the most important behaviors is \emph{consensus}, where a group of interconnected agents dynamically achieves a common agreement. The canonical model for this problem is the linear Laplacian dynamics, a subject of extensive research for decades \cite{olfati-saberConsensusCooperationNetworked2007, jadbabaieCoordinationGroupsMobile2003, ren2005,blondelConvergenceMultiagentCoordination2005}. The elegance of this linear model lies in its direct connection to the spectral properties of the underlying network graph. For a connected network of agents with linear interactions, reaching a consensus is guaranteed, and the agents \emph{fully synchronize} to a common constant value.

Ubiquitous in real-world systems, the emergence of stable opinion clusters and persistent disagreement, rather than full synchronization, necessitates models that extend beyond the linear paradigm. Prior research has primarily examined two mechanisms for breaking consensus. The first, prominent in social dynamics, involves a \emph{dynamic topology}, where agents cut communication with those who have disparate views, as seen in bounded confidence models \cite{ deffuantMixingBeliefsInteracting2000,krause2002,MG10}. The second mechanism retains a fixed topology but introduces \emph{antagonistic interactions}, where repulsion between agents on a signed graph can lead to polarized states \cite{altafiniConsensusProblemsNetworks2013, hendrickxLiftingApproachModels2014}.

This paper advances a third, rapidly developing line of research: the emergence of disagreement on a \emph{fixed, cooperative network}, driven purely by the \emph{inherent nonlinearity} of agent interactions: a perception-based influence. Unlike antagonistic or adaptive topologies models, this framework explains persistent disagreement as a property of the interaction law itself. 
Prior work has explored two complementary manifestations of such nonlinearities: those arising from \emph{quantized communication}, where discrete actions replace continuous opinions \cite{Martins2008,chowdhuryContinuousOpinionsDiscrete2016, ceragioliConsensusDisagreementRole2018}, and those based on \emph{smooth sigmoidal activation} inspired by collective decision-making in biological and social systems \cite{grayMultiagentDecisionMakingDynamics2018, bizyaevaNonlinearOpinionDynamics2023,baumannModelingEchoChambers2020}. Seminal theoretical work has rigorously linked the emergence of multiple equilibria to the network's algebraic connectivity for specific sigmoidal models \cite{fontanMultiequilibriaAnalysisClass}. However, a unified theory capturing the transition from consensus to clustering across a broad class of nonlinearities, encompassing both smooth and quantized interactions, remains lacking.  Furthermore, a formal framework for quantifying the \emph{robustness} of the resulting clusters with respect to the perturbation from the rest of the network is absent from the literature.

To this end, we investigate the fundamental trade-off between network connectivity and interaction nonlinearity to: establish explicit conditions for full synchronization, characterize the equilibria when these conditions are violated, and formally quantify the robustness of these emergent structures against network perturbations. Our main contributions are:

\begin{itemize}
    \item \emph{Sharp Threshold for Full Synchronization:} We establish a sufficient condition that guarantees full synchronization. This result provides a direct link between the agents' maximal interaction nonlinearity (Lipschitz constant) and the network's algebraic connectivity. We then prove constructively that this threshold is sharp: its violation is a sufficient condition for the \emph{existence of an admissible signal function} that generates disagreement.
    \item \emph{Characterization of Emergent Clustered Equilibria:} We provide necessary conditions on the signal function itself for the emergence of structured disagreement. This reveals that, independent of the network topology, clustering is only possible if the interaction law possesses specific features, such as \emph{unstable fixed points}.
    
    \item \emph{Formal Framework for Cluster Robustness:} We develop a general methodology proving that these clustered equilibria are ISS with respect to network perturbations. This allows us to prove a fundamental result: cluster cohesion is sensitive to \emph{heterogeneity} and not to the \emph{magnitude} of external influence.
\end{itemize}

Centering the approach on the interaction law itself and its fixed points provides a unifying framework for previously distinct lines of research. It treats systems with smooth sigmoidal interactions as a direct special case, while capturing those with quantized communication as a limiting case of steep, continuous approximations. Furthermore, our application of the ISS framework provides, for the first time, a formal methodology for quantifying the \emph{robustness} of the clustered equilibria that emerge. This allows an analysis that goes beyond solely the existence of equilibria to characterize their stability from the rest of the network.

The remainder of this paper is organized as follows. Section~\ref{sec:framework} introduces the system model and our core definitions. Section~\ref{sec:qualitative} establishes the fundamental qualitative properties of the dynamics. Section~\ref{sec:equilibria} provides a detailed characterization of both synchronized and clustered equilibria. Section~\ref{sec:Threshold_global_sync} derives the sharp threshold for full synchronization. Section~\ref{sec:ISS} develops the Input-to-State Stability framework for robust clusters, including a validating numerical example. Section~\ref{sec:conclusion} concludes the paper and discusses future directions.

\textbf{Notation} In the following, we denote by $\R$ the set of real numbers. For a vector $\bm{x} \in \R^N$, we denote by $x_i$ the $i$-th component of $\bm{x}$. For a matrix $\bm{A} \in \R^{N\times N}$, we denote by $a_{ij}$ the element of $\bm{A}$ at the $i$-th row and $j$-th column. We denote by $\zero$ and $\one$ the vector of $\R^N$ with all components equal to $0$ and $1$, respectively. For a column vector $\bm{x} \in \R^N$, we denote by $\norm{\bm{x}} = (\bm{x}^\top \bm{x})^{\frac{1}{2}}$ the Euclidean norm of $\bm{x}$ and, for a positive definite matrix $\bm{D}$, by $\norm{\bm{x}}_{\bm{D}} = (\bm{x}^\top \bm{D} \bm{x})^{\frac{1}{2}}$ the norm induced by $\bm{D}$. 
Moreover, we denote by $\mathrm{diag}(\bm{x}) \in \R^{N\times N}$ the diagonal matrix with diagonal elements given by the vector $\bm{x} \in \R^{N}$. For two vectors $\bm{x},\bm{y} \in \R^N$, $\bm{x} \leq \bm{y}$ means $x_i \leq y_j$ for all $i \in \{1,\dots,N\}$, if, in addition $\bm{x}\neq \bm{y}$, we note $\bm{x} < \bm{y}$. For a function $f: \mathcal{X} \to \mathcal{X}$, the set $\mathrm{Fix}(f) = \left\{ x \in \mathcal{X} \mid x =  f(x) \right\}$ contains the fixed points of $f$.

\section{A Framework for Consensus Dynamics with Nonlinear Interactions} \label{sec:framework}

\subsection{Graph Structure and Network Interaction}

Consider the set $\Vcal = \{1, \dots, N\}$ of $N$ agents or vertices interacting with each other. The interactions are described by a graph $\Gcal = \left(\Vcal, \Ecal\right)$, where the edge set $\Ecal\subset\Vcal\times\Vcal$ indicates whether two agents are interacting. The fundamental properties and assumptions of this interaction graph are as follows.

\begin{standassumption}\label{ass:graph}
    The graph $\Gcal$ is \textbf{time-invariant, undirected, connected, and simple} (i.e., it has no self-loops or multiple edges).
\end{standassumption}

The topology of the network is captured by its symmetric \emph{adjacency matrix}, $\bm{A} \in \R^{N\times N}$, where $a_{ij} = 1$ if $(i,j) \in \Ecal$ and $a_{ij} = 0$ otherwise. From this, we define two other key matrices: the \emph{degree matrix}, $\bm{D} := \mathrm{diag}(d_1, \dots, d_N) \in \R^{N\times N} $, where $d_i = \sum_{j=1}^{N} a_{ij}$ is the degree of vertex $i$, and the \emph{Laplacian matrix}, $\bm{L} := \bm{D} - \bm{A} \in \R^{N\times N}$. The set of direct neighbors of agent $i$ is its \emph{neighborhood}: $\mathcal{N}_i = \{ j \in \Vcal \mid (i,j) \in \Ecal \}$, with cardinality $d_i$.

An induced subgraph of $\Gcal$ is a graph $\Gcal' = \left(\Vcal', \Ecal'\right)$ such that $\Vcal' \subset \Vcal$ and $\Ecal' = \Ecal \cap \left(\Vcal'\times \Vcal'\right)$ is the set of \emph{internal edges}. Its topology is described by the induced adjacency matrix $\bm{A}_{\mathrm{in}} \in \mathbb{R}^{N' \times N'}$, where $N' = |\Vcal'|$.

A crucial consequence of Assumption~\ref{ass:graph} is that the adjacency matrix $\bm{A}$ is \emph{irreducible} due to the graph being connected. This allows us to apply the Perron-Frobenius theorem, providing foundational results on the spectral properties of such matrices.

\begin{lemma}[Perron-Frobenius {\cite[Theorem~2.12]{bulloLecturesNetworkSystems2018}}] \label{lemma:perron_frobenius}
	Let $\bm{A} \in \R^{N \times N}$ with $N \geq 2$. If $\bm{A}$ is a non-negative matrix, then there exists a real eigenvalue $\lambda \geq \abs{\mu} \geq 0$ for all other eigenvalues $\mu$ of $\bm{A}$ and the right and left eigenvectors $\bm{v}$ and $\bm{w}$ of $\bm{A}$ associated with the eigenvalue $\lambda$ are non-negative. i.e., $\bm{v} \geq \bm{0}$ and $\bm{w} \geq \bm{0}$. 
	
	If additionally $\bm{A}$ is irreducible, then the eigenvalue $\lambda$ is strictly positive and simple, and the right and left eigenvectors $\bm{v}$ and $\bm{w}$ of $\bm{A}$ associated with the eigenvalue $\lambda$ are unique and positive, up to rescaling. i.e., $\bm{v} > \bm{0}$ and $\bm{w} > \bm{0}$.
\end{lemma}

\subsection{Agent and Collective Dynamics}

To each agent $i$ we assign a normalized scalar state $x_i \in \left[-1,1\right]$ (e.g., an opinion). Agents are assumed to interact based on their perception of the neighbors' states, a process modeled by a \textbf{common signal function} $s: \left[-1,1\right] \to \left[-1,1\right]$. This function acts as the lens through which agents see each other's states before the averaging process.

The driving mechanism for the agent's evolution is the discrepancy between the current state and the average of the perceived state of the neighbors. This consensus-seeking behavior is captured by the following feedback law:
\begin{equation}\label{eq:dynamic_single_agent}
	\dot{x}_i = \frac{1}{d_i} \sum_{j = 1}^{N} a_{ij} s(x_j) - x_i.
\end{equation}

The right-hand side of this equation represents a \emph{disagreement vector} for agent $i$. Its state velocity, $\dot{x}_i$, is directly proportional to the difference between what it ``sees" from its neighborhood and its own internal state. A positive disagreement pulls its state up, while a negative disagreement pushes it down, thus naturally driving the agent towards the local average.

The collective dynamics of all $N$ agents can be expressed in the compact collective vector form:
\begin{equation}\label{eq:dynamic}
	\dot{\bm{x}} = \bm{D}^{-1} \! \bm{A} \bm{s}(\bm{x}) - \bm{x} := \bm{f}(\bm{x}),
\end{equation}
where $\bm{x} = (x_1,\dots,x_N)^\top$ and $\bm{s}(\bm{x}) = (s(x_1),\dots,s(x_N))^\top$ belong to the hypercube $\mathcal{X} \coloneq \left[-1,1\right]^N$, which defines the system state space.

The choice of the row-stochastic matrix $\bm{D}^{-1}\bm{A}$, often called the \emph{random-walk normalized} adjacency matrix, is a crucial modeling decision. Unlike dynamics governed by the unweighted graph Laplacian ($\bm{L}=\bm{D}-\bm{A}$), this formulation ensures that each agent updates its state based on the \emph{average} of the signals received from its neighbors, rather than their sum. Consequently, the influence of a neighborhood is independent of its size, preventing high-degree nodes from having an overly fast response. This modeling paradigm is representative of many real-world processes, such as social opinion formation or distributed estimation, where an individual is influenced by the prevailing sentiment of its local group, not the sheer volume of its connections.

\subsection{Signal Function and Communication Model}

The nonlinearity introduced by $s(x)$ allows us to model various interaction scenarios. For instance, by choosing $s(x)=x$, one recovers the normalized linear consensus model, while an affine function $s(x)=ax+b$ can model biased consensus algorithms. Sigmoidal functions, such as $s(x) = \tanh(x)$, inspired by biological systems, can also represent saturation or decoding effects in communication \cite{grayMultiagentDecisionMakingDynamics2018,bizyaevaNonlinearOpinionDynamics2023}.

In the following, we will restrict our study to the signal function, verifying the following assumption:

\begin{standassumption}\label{ass:signal}
    The function $s: [-1,1] \to [-1,1]$ is \textbf{non-decreasing} and \textbf{Lipschitz-continuous} with constant $K$\!. 
\end{standassumption}

The assumption that the signal function $s$ is non-decreasing is central to our analysis, providing justification from both a modeling and a mathematical standpoint. From a modeling perspective, it captures a natural causality inherent to many physical, biological, and social systems: a higher internal state (e.g., opinion, temperature, activation) should induce a correspondingly higher, or at least non-lower, signal to its neighbors. 

Mathematically, this non-decreasing property is precisely what ensures that the influence between any two connected agents is fundamentally \emph{cooperative} \cite{hirschChapter4Monotone2006}. As we will formally establish in Section~\ref{sec:qualitative}, this is the key that unlocks the rigorous analytical framework of monotone dynamical systems, which allows us to prevent complex oscillatory behaviors and guarantee convergence to a set of equilibria.

The fully synchronized equilibria of the system occur when all agents reach a common state $c$ such that $c = s(c)$; in other words, they are precisely the fixed points of the signal function, $c \in \mathrm{Fix}(s)$. As we will see later, the stability of such an equilibrium is determined by its local behavior on the fully synchronized manifold, where the dynamics reduces to $\dot{x} = s(x) - x$. The term $s(x)-x$ acts as a restoring force when it pulls the state towards the fixed point, and a repelling force when it pushes the state away. To formalize this crucial behavior, we classify the fixed points based on their local stability properties.
\begin{definition}\label{def:fixed_point_stability}
    A signal function $s$ is said to be:
    \begin{itemize}
		\item an \emph{underestimation} if $x(s(x) - x) \leq 0$ for all $x \in \left[-1,1\right]$.
        \item an \emph{overestimation} if $x(s(x) - x) \geq 0$ for all $x \in \left[-1,1\right]$.
    \end{itemize}

	Furthermore, a fixed point $c \in \mathrm{Fix}(s)$ is classified based on its one-sided stability:
    \begin{itemize}
        \item It is \emph{left-stable} (resp. \emph{right-stable}) if there exists a neighborhood $I \subseteq [-1,1]$ of $c$ such that $(x-c)(s(x)-x)\leq 0$ for all $x \in I$ with $x<c$ (resp. $x>c$).
        \item It is \emph{left-unstable} (resp. \emph{right-unstable}) if the strict inequality $(x-c)(s(x)-x)> 0$ holds under the same conditions.
    \end{itemize}
    Based on these properties, we define a fixed point as \emph{stable} if it is both left- and right-stable, \emph{unstable} if it is both left- and right-unstable, and \emph{semi-stable} otherwise.
\end{definition}

The local stability defined in Definition~\ref{def:fixed_point_stability} captures a fundamental dichotomy. A \emph{stable fixed point} acts as a local attractor: for a state $x$ slightly deviating from $c$, the dynamics generates a restoring force that pushes $x$ back towards $c$. In contrast, an \emph{unstable} or \emph{semi-stable fixed point} is locally repelling from at least one side. As we will prove, this simple one-dimensional classification is the key determinant of stability for the full $N$-dimensional synchronized states.

Based on this property, we partition the set of all fixed points, $\mathrm{Fix}(s)$, into the set of stable fixed points, which we denote by $\mathrm{Fix}^{\bullet}(s)$, and the set of unstable or semi-stable ones, $\mathrm{Fix}^{\circ}(s)$. For a more granular analysis, we will also need to identify the specific sources of instability. We define $\mathrm{Fix_L^{\circ}}(s)$ as the set of all fixed points that are \emph{unstable from the left}, and $\mathrm{Fix_R^{\circ}}(s)$ as the set of those that are \emph{unstable from the right}. Note that a fully unstable fixed point belongs to both of these subsets, as shown in the following fundamental lemma.


\begin{lemma}\label{lemma:existence_of_stable_fixed_point}
	Any continuous signal function $s$ has at least one stable fixed point. i.e., $\mathrm{Fix}^{\bullet}(s) \neq \emptyset$.
\end{lemma}
\begin{proof}
    See Appendix~\ref{app:proof_lemma2}.
\end{proof}

\section{Qualitative Behavior: Invariance and Convergence}\label{sec:qualitative}

Despite its high dimensionality and nonlinear nature, the system \eqref{eq:dynamic} possesses remarkable properties that make its asymptotic behavior predictable. This section is dedicated to establishing two useful fundamental properties that can tackle this complexity. First, by proving that the system is \emph{cooperative}, we rule out complex oscillatory behaviors and guarantee that almost all trajectories converge towards the set of equilibria. Second, by identifying key \emph{invariant sets}, we show that these trajectories are confined to specific, well-defined regions of the state space. Together, these properties allow us to shift our focus from asking \emph{if} the system settles down, to analyzing precisely \emph{where} it settles.

As a direct consequence of the dynamics, the hypercube state space $\mathcal{X} = [-1,1]^N$ is forward invariant \cite{couthuresGlobalSynchronizationMultiagent2025}, ensuring solutions remain bounded. A primary invariant set of interest is the \emph{Fully Synchronized Manifold (FSM)} denoted $\Scal$, which represents the subspace of complete agent agreement:
\begin{equation*}
    \mathcal{S} := \left\{ \bm{x} \in \mathcal{X} \mid x_i = x_j,\,  \forall i, j \in \Vcal \right\} = \mathrm{Span}(\one) \cap \mathcal{X}.
\end{equation*}
For this class of systems, the FSM is also forward invariant \cite{couthuresGlobalSynchronizationMultiagent2025}, implying that a trajectory starting on the manifold remains on it for all future times.

\subsection{Monotonicity and Guaranteed Convergence to Equilibria}

The cooperative nature of the system arises from the non-decreasing signal function $s$, which ensures a positive influence between connected agents ($\bm{A} \geq 0$). In such systems, trajectories exhibit strong ordering properties; for instance, a trajectory starting coordinate-wise below another will remain so forever. This structure is key to preventing oscillations and chaos. For a comprehensive survey, see \cite{hirschDynamicalSystemsApproach1984,hirschChapter4Monotone2006}.

\begin{definition}
    A dynamical system $\dot{\bm{x}} = \bm{f}(\bm{x})$ in $\mathcal{X}$ is \emph{cooperative} if for any initial conditions $\bm{x}(0) \leq \bm{y}(0)$ in $\mathcal{X}$, the resulting trajectories satisfy $\bm{x}(t) \leq \bm{y}(t)$ for all $t \geq 0$.
\end{definition}

\begin{proposition}\label{prop:monotone_flow}
    The dynamical system \eqref{eq:dynamic} is cooperative on $\mathcal{X}$. Consequently, it does not admit attracting cycles, and almost all trajectories converge to the set of equilibria. 
\end{proposition}
\begin{proof}
    To show that \eqref{eq:dynamic} is cooperative on $\mathcal{X}$, since $\mathcal{X}$ is forward invariant, it is sufficient to show that $\bm{f}: \mathcal{X} \to \mathcal{X}$ is a quasi-monotone function i.e., for any $\bm{x}, \bm{y} \in \mathcal{X}$, such that $\bm{x} \leq \bm{y}$ one has $f_i(\bm{x}) \leq f_i(\bm{y})$ for all $i \in \Vcal$ such that $x_i = y_i$, where $f_i$ represents the $i$-th component of function $\bm{f}$. This is known as the Kamke--Müller condition \cite[Theorem 3.2]{hirschChapter4Monotone2006}.

    Let $\bm{x}, \bm{y} \in \mathcal{X}$ satisfying the conditions $\bm{x} \leq \bm{y}$ and $x_i = y_i$ for all $i \in \mathcal{I} \subset \Vcal$. For any $i \in \mathcal{I}$, one has that
    \begin{align*}
        f_i(\bm{y}) - f_i(\bm{x}) &= \frac{1}{d_i} \sum_{j=1}^{N} a_{ij} \left( s(y_j) - s(x_j) \right) - \left( y_i - x_i \right)\\
        &= \frac{1}{d_i} \sum_{j=1}^{N} a_{ij} \left( s(y_j) - s(x_j) \right) \geq 0.
    \end{align*}
    Since for all $i,j \in \Vcal$, $d_i >0$, $a_{ij} \geq 0$ and $s$ is non-decreasing. Therefore, $\bm{f}$ is a quasi-monotone function on $\mathcal{X}$ and the system \eqref{eq:dynamic} is cooperative. Finally, from \cite[Corollary~2.4]{hirschDynamicalSystemsApproach1984} we have that cooperative systems do not admit attracting cycles and almost all trajectories converge to an equilibrium.
\end{proof}

Proposition~\ref{prop:monotone_flow} is a key result, ensuring that almost all the trajectories of the system will eventually settle at a steady state. Under the regularity conditions on the signal function, the result can be extended to all trajectories as follows.

\begin{lemma}\label{lemma:SOP_flow}
    Let $s$ be a strictly increasing and continuously differentiable function. Then, any trajectory of \eqref{eq:dynamic} will asymptotically converge to an equilibrium.  
\end{lemma}

\begin{proof}
    For a smooth cooperative system, it is sufficient to show that the Jacobian matrix of the vector field is irreducible \cite[Chapter~4, Theorem~1.1]{smithMonotoneDynamicalSystems2008}. The Jacobian of \eqref{eq:dynamic} is $\bm{J}(\bm{x}) = \bm{D}^{-1}\bm{A}\mathrm{diag}(\bm{s'}(\bm{x})) - \bm{I}$. Since $s$ is strictly increasing, all entries of the diagonal matrix $\mathrm{diag}(\bm{s'}(\bm{x}))$ are positive. The irreducibility of $\bm{J}(\bm{x})$ is determined by the pattern of its off-diagonal entries, which is identical to that of $\bm{D}^{-1}\bm{A}$. As the graph is connected, $\bm{D}^{-1}\bm{A}$ is irreducible, and thus, $\bm{J}(\bm{x})$ as well. Convergence to an equilibrium then follows from the boundedness of trajectories \cite[Proposition~1.2 and Corollary~1.9]{hirschChapter4Monotone2006}.
\end{proof}

\subsection{Bounding Trajectories with Invariant Hypercubes} \label{subsec:invariant_sets}

Proposition~\ref{prop:monotone_flow} guarantees that trajectories settle at an equilibrium, but it does not specify to which equilibrium when multiple equilibria exist. This section aims to provide further insights into this issue. The trajectories are constrained by invariant hypercubes generated by the fixed points of the signal function $s$. These ``trapping regions" are essential for localizing equilibria and analyzing their basins of attraction. The following proposition formalizes this idea and is illustrated in Figure~\ref{fig:invariance_of_hypercubes_defined_by_fixed_points}.

\begin{proposition}\label{prop:invariance_of_hypercubes_defined_by_fixed_points}
    The following statements hold true for dynamics \eqref{eq:dynamic}:
    \begin{enumerate}
        \item For any $a, b \in [-1, 1]$ with $a < b$, if $s(a) \geq a$ and $s(b) \leq b$, the hypercube $\left[a, b\right]^N$ is forward invariant. \label{prop:invariance_of_hypercubes_defined_by_fixed_points:item:1}
        \item For any $a, b \in [-1, 1]$ with $a < b$, if $s(a) > a$ and $s(b) < b$, the hypercube $\left(a, b\right)^N$ is forward invariant. \label{prop:invariance_of_hypercubes_defined_by_fixed_points:item:2}
    \end{enumerate}
\end{proposition}

\begin{proof}
    \ref{prop:invariance_of_hypercubes_defined_by_fixed_points:item:1} Since the system \eqref{eq:dynamic} admits a unique maximal trajectory for any initial condition (since $\bm{f}$ is Lipschitz continuous), by Nagumo's theorem (see \cite{blanchiniSetInvarianceControl1999}), it is sufficient to analyze the dynamics at the boundary of the hypercube $[a,b]^N$.

	Consider any agent $i \in \Vcal$. Suppose its state reaches the upper boundary, $x_i = b$. The states of all other agents satisfy $x_j \leq b$. Since $s$ is non-decreasing, $s(x_j) \leq s(b)$. The dynamics for agent $i$ is:
	\begin{align*}
		\dot{x}_i = \frac{1}{d_i} \sum_{j=1}^N a_{ij} s(x_j) - b \leq \frac{1}{d_i} \sum_{j=1}^N a_{ij} s(b) - b.
	\end{align*}
	Since the matrix $\bm{D}^{-1}\bm{A}$ is row-stochastic, $\sum_{j=1}^N a_{ij}/d_i = 1$. This simplifies the inequality to $\dot{x}_i \leq s(b) - b$. By assumption, $s(b) \leq b$, so $\dot{x}_i(t) \leq 0$. This ensures that the trajectory does not exit through the upper face $x_i = b$.

	By symmetry, using the same argument, trajectories cannot exit through the lower face $x_i = a$.
	Since this holds for any agent $i$, the hypercube $[a, b]^N$ is forward invariant. 

	\ref{prop:invariance_of_hypercubes_defined_by_fixed_points:item:2} Since $s$ is non-decreasing, there exists a $\varepsilon > 0$ such that for all $y \in [b - \varepsilon, b]$, we have $s(y) < y$ and for all $y \in [a, a + \varepsilon]$, we have $s(y) > y$. Let us note $H_n = [a + \varepsilon/n, b - \varepsilon/n]^N$ for $n \geq 1$. By \ref{prop:invariance_of_hypercubes_defined_by_fixed_points:item:1}, for $n \geq 1$, the hypercube $H_n$ is forward invariant and the sequence $H_n$ is increasing. i.e., $H_n \subset H_{n+1}$. Since every $H_n$ is forward invariant, so is their union:
	\begin{equation*}
		\bigcup_{n \geq 1} H_n = \bigcup_{n \geq 1} [a + \frac{\varepsilon}{n}, b - \frac{\varepsilon}{n}]^N = (a, b)^N.
	\end{equation*}
	Thus, the open hypercube $(a, b)^N$ is forward invariant.	
    \end{proof}
\begin{corollary}\label{cor:invariance_of_hypercubes_defined_by_fixed_points}
Proposition~\ref{prop:invariance_of_hypercubes_defined_by_fixed_points} can be applied to the fixed points of $s$ as follows:
    \begin{enumerate}
        \item For any $\underline{c}, \overline{c} \in \mathrm{Fix}(s)$ such that $\underline{c} < \overline{c}$, the hypercube $\left[\underline{c}, \overline{c}\right]^N$ is forward invariant. \label{prop:invariance_of_hypercubes_defined_by_fixed_points:item:1bis}
        \item In addition, if $\underline{c} \in \mathrm{Fix_R^{\circ}}(s)$ and $\overline{c} \in \mathrm{Fix_L^{\circ}}(s)$, then the hypercube $\left(\underline{c}, \overline{c}\right)^N$ is forward invariant. \label{prop:invariance_of_hypercubes_defined_by_fixed_points:item:2bis}
    \end{enumerate}
     Consequently, if $\underline{c}, \overline{c} \in \mathrm{Fix}(s)$ such that $\underline{c} \leq \min_{i \in \Vcal} x_i(0)$ and $\max_{i \in \Vcal} x_i(0) \leq \overline{c} $, then $\bm{x}(t) \in \left[\underline{c}, \overline{c}\right]^N$, for all $t \geq 0$.
\end{corollary}
\begin{proof}Any fixed point $c$ satisfies $s(c)=c$ thus Proposition \ref{prop:invariance_of_hypercubes_defined_by_fixed_points} applies to $\underline{c}, \overline{c}$.
\end{proof}

\begin{figure}
	\centering
    \begin{subfigure}[t]{0.48\linewidth}
    	\centering
    	\begin{tikzpicture}[scale=0.65,
    			vec/.style={
    					decorate,
    					decoration={
    							markings,      
    							mark=at position 0.99 with {\arrow{stealth}}
    						}}]
    
    		\pgfmathsetmacro{\cUnder}{-0.5}
    		\pgfmathsetmacro{\cOver}{0.7}
    		\pgfmathsetmacro{\aVal}{-0.75} 
    		\pgfmathsetmacro{\aValprime}{0.1} 
    		\pgfmathsetmacro{\bVal}{0.6} 
    		\pgfmathsetmacro{\bValprime}{0.5} 
    
    		\begin{axis}[
    				width=8cm, height=8cm,
    				axis lines=middle,
    				xlabel={$x$},
    				ylabel={$y$},
    				xmin=-1.1, xmax=1.1,
    				ymin=-1.1, ymax=1.1,
    				xtick={-1, 0, 1},
    				ytick={-1, 0, 1},
    				extra x ticks={\aVal, \aValprime, \bVal, \bValprime},
    				extra x tick labels={$a$, $a'$, $b$, $b'$},
    				extra y ticks={\aVal, \aValprime, \bVal, \bValprime},
    				extra y tick labels={$a$, $a'$, $b$, $b'$},
    				xticklabel style={font=\small},
    				yticklabel style={font=\small, anchor=east},
    				typeset ticklabels with strut,
    				grid=major,
    				legend pos=north west,
    			]

    			\coordinate (P1) at (axis cs: -1.0, -0.7);
    			\coordinate (P2) at (axis cs: -0.7, -0.7);
    			\coordinate (P3) at (axis cs: -0.6, -0.7);
    			\coordinate (P4) at (axis cs: \cUnder, \cUnder);
    			\coordinate (P5) at (axis cs: -0.4, 0.2);
    			\coordinate (P6) at (axis cs: 0.2, 0.3);
    			\coordinate (P7) at (axis cs: 0.3, 0.3);
    			\coordinate (P8) at (axis cs: 0.65, 0.5);
    			\coordinate (P9) at (axis cs: \cOver, \cOver);
    			\coordinate (P10) at (axis cs: 0.75, 0.85);
    			\coordinate (P11) at (axis cs: 0.85, 0.85);
    			\coordinate (P12) at (axis cs: 1.0, 0.85);

    			\addplot[domain=-1:1, samples=2, dashed, color=gray, thick] {x};
    			\addlegendentry{\small $y=x$}			
    			\addlegendimage{fill=myblue, draw=myblue, thick}
    			\addlegendentry{\small $s(x)$}

    			\addlegendimage{only marks, mark=o, fill=white, draw=black}
    			\addlegendentry{\small $\mathrm{Fix}^{\circ}(s)$}
    			\addlegendimage{only marks, mark=*, fill=black, draw=black}
    			\addlegendentry{\small $\mathrm{Fix}^{\bullet}(s)$}
    
    			\draw[myblue, very thick] (P1) -- (P2) -- (P3) -- (P4) -- (P5) -- (P6) -- (P7) -- (P8) -- (P9) -- (P10) -- (P11) -- (P12);
    			\node[circle, draw, fill=white, inner sep=1.5pt] at (axis cs: \cUnder, \cUnder) {};
    			\node[circle, draw, fill=white, inner sep=1.5pt] at (axis cs: \cOver, \cOver) {};
    			\node[font=\small, anchor=south east] at (axis cs: \cUnder, \cUnder) {$c_2$};
    			\node[font=\small, anchor=north west] at (axis cs: \cOver, \cOver) {$c_4$};

    			\node[circle, draw, fill=black, inner sep=1.5pt] at (axis cs: -0.7, -0.7) {};
    			\node[font=\small, anchor=north west, color=black] at (axis cs: -0.7, -0.7) {$c_1$};
    
    			\node[circle, draw, fill=black, inner sep=1.5pt] at (axis cs: 0.3, 0.3) {};
    			\node[font=\small, anchor=north west, color=black] at (axis cs: 0.3, 0.3) {$c_3$};
    
    			\node[circle, draw, fill=black, inner sep=1.5pt] at (axis cs: 0.85, 0.85) {};
    			\node[font=\small, anchor=north west, color=black] at (axis cs: 0.85, 0.85) {$c_5$};
    		\end{axis}
    	\end{tikzpicture}
    	\caption{A piecewise linear signal function $s(x)$ with two unstable fixed points $c_2$ and $c_4$ and three stable fixed points $c_1$, $c_3$, and $c_5$. The shape of $s$ indicates the presence of induced invariant hypercubes.}
    	\label{subfig:signal_function_example}
    \end{subfigure}
    \hfil
    \begin{subfigure}[t]{0.48\linewidth}
    	\begin{tikzpicture}[scale=0.65,
    			vec/.style={
    					decorate,
    					decoration={
    							markings,      
    							mark=at position 0.99 with {\arrow{stealth}}
    						}}]
    
    		\pgfmathsetmacro{\cUnder}{-0.5}
    		\pgfmathsetmacro{\cOver}{0.7}
    
    		\pgfmathsetmacro{\aVal}{-0.75} 
    		\pgfmathsetmacro{\aValprime}{0.1} 
    		\pgfmathsetmacro{\bVal}{0.6} 
    		\pgfmathsetmacro{\bValprime}{0.5} 
     \pgfplotsset{grid style={help lines}} 
    		\begin{axis}[
    				width=8cm, height=8cm,
    				area legend,
    				axis lines=middle,
    				xlabel={$x_1$},
    				ylabel={$x_2$},
    				xmin=-1.1, xmax=1.1,
    				ymin=-1.1, ymax=1.1,
    				xtick={-1, 0, 1},
    				ytick={-1, 0, 1},
    				extra x ticks={\aVal, \aValprime, \bVal, \bValprime},
    				extra x tick labels={$a$, $a'$, $b$, $b'$},
    				extra y ticks={\aVal, \aValprime, \bVal, \bValprime},
    				extra y tick labels={$a$, $a'$, $b$, $b'$},
    				xticklabel style={font=\small},
    				yticklabel style={font=\small, anchor=east},
    				typeset ticklabels with strut,
    				grid=major,
    				legend pos=south east,
    				legend columns=1,
    				legend cell align=left,
                    axis on top=false,
                    grid style={/pgfplots/on layer=axis background}
    			]
    
    			\begin{pgfonlayer}{background}
    				\addplot[
    					fill=mygreen!30, fill opacity=0.7,
    					draw=mygreen, thick,
                        pattern=north east lines,
    					forget plot 
    				] (\cUnder, \cUnder) -- (\cOver, \cUnder) -- (\cOver, \cOver) -- (\cUnder, \cOver) -- cycle;
    
    				\addplot[
    					fill=mygreen!30, fill opacity=0.7,
    					draw=mygreen, thick,
                        pattern=north east lines,
    				] (\cOver, \cOver) -- (1, \cOver) -- (1, 1) -- (\cOver, 1) -- cycle;
    				\addlegendentry{\small $[\underline{c}, \overline{c}]^2$}		
    				\addplot[
    					fill=mygreen!30, fill opacity=0.7,
    					draw=mygreen, thick,
                        pattern=north east lines,
    					forget plot 
    				] (\cUnder, \cUnder) -- (-1, \cUnder) -- (-1, -1) -- (\cUnder, -1) -- cycle;        

    				\addplot[
    					fill=myblue!30, fill opacity=0.7,
    					draw=myblue, thick,
    				] (\aVal, \aVal) -- (\bVal, \aVal) -- (\bVal, \bVal) -- (\aVal, \bVal) -- cycle;
    				\addlegendentry{\small $[a,b]^2$}
    
    				\addplot[
    					fill=myred!30, fill opacity=0.7,
    					draw=myred, thick,
    				] (\aValprime, \aValprime) -- (\bValprime, \aValprime) -- (\bValprime, \bValprime) -- (\aValprime, \bValprime) -- cycle;
    				\addlegendentry{\small $[a',b']^2$}
    			\end{pgfonlayer}
    
    			\draw[color=black, thick] (-1, -1) -- (1, 1);
    			\node[font=\small, anchor=north west, color=black] at (axis cs: 1, 1) {$\!\mathcal{S}$};
    
    			\node[circle, draw, fill=white, inner sep=1.5pt] at (axis cs: \cUnder, \cUnder) {};
    			\node[circle, draw, fill=white, inner sep=1.5pt] at (axis cs: \cOver, \cOver) {};
    			\node[font=\small, anchor=south east] at (axis cs: \cUnder, \cUnder) {$c_2 \bm{1}$};
    			\node[font=\small, anchor=north west] at (axis cs: \cOver, \cOver) {$c_4 \bm{1}$};
    			\pgfmathsetmacro{\numArrows}{2}
    			\pgfmathsetmacro{\step}{(\cOver - \cUnder) / (\numArrows + 1)}
    			\pgfmathsetmacro{\arrowLength}{0.06} 
    
    			\pgfmathsetmacro{\numArrowssmall}{1}
    			\pgfmathsetmacro{\stepssmall}{(1 - \cOver) / (\numArrowssmall + 1)}
    			\pgfmathsetmacro{\arrowLengthssmall}{0.06} 
    
    			\pgfmathsetmacro{\numArrowssmallprime}{1}
    			\pgfmathsetmacro{\stepssmallprime}{(\cUnder - (-1)) / (\numArrowssmallprime + 1)}
    			\pgfmathsetmacro{\arrowLengthssmallprime}{0.06} 

    			\pgfmathsetmacro{\numArrowsab}{2}
    			\pgfmathsetmacro{\stepsab}{(\bVal - \aVal) / (\numArrowsab + 1)}
    
    			\pgfmathsetmacro{\numArrowsabprime}{1}
    			\pgfmathsetmacro{\stepsabprime}{(\bValprime - \aValprime) / (\numArrowsabprime + 1)}
    			\pgfmathsetmacro{\arrowLengthab}{0.06} 
    
    			\begin{scope}[>=Triangle]
    				\pgfplotsinvokeforeach{1,...,\numArrows}{
    					\draw[->, color=mygreen, thick] (axis cs: \cOver, {\cUnder + #1 * \step}) -- (axis cs: {\cOver - \arrowLength}, {\cUnder + #1 * \step});
    					\draw[->, color=mygreen, thick] (axis cs: \cUnder, {\cUnder + #1 * \step}) -- (axis cs: {\cUnder + \arrowLength}, {\cUnder + #1 * \step});
    					\draw[->, color=mygreen, thick] (axis cs: {\cUnder + #1 * \step}, \cOver) -- (axis cs: {\cUnder + #1 * \step}, \cOver - \arrowLength);
    					\draw[->, color=mygreen, thick] (axis cs: {\cUnder + #1 * \step}, \cUnder) -- (axis cs: {\cUnder + #1 * \step}, {\cUnder + \arrowLength});
    
    					\draw[->, color=myblue, thick] (axis cs: \bVal, {\aVal + #1 * \stepsab}) -- (axis cs: {\bVal - \arrowLength}, {\aVal + #1 * \stepsab});
    					\draw[->, color=myblue, thick] (axis cs: \aVal, {\aVal + #1 * \stepsab}) -- (axis cs: {\aVal + \arrowLength}, {\aVal + #1 * \stepsab});
    					\draw[->, color=myblue, thick] (axis cs: {\aVal + #1 * \stepsab}, \bVal) -- (axis cs: {\aVal + #1 * \stepsab}, \bVal - \arrowLength);
    					\draw[->, color=myblue, thick] (axis cs: {\aVal + #1 * \stepsab}, \aVal) -- (axis cs: {\aVal + #1 * \stepsab}, {\aVal + \arrowLength});
    
    				}

    				\pgfplotsinvokeforeach{1,...,\numArrowssmall}{
    					\draw[->, thick, color=mygreen] (axis cs: \cOver, {\cOver + #1 * \stepssmall}) -- (axis cs: {\cOver + \arrowLengthssmall}, {\cOver + #1 * \stepssmall});
    					\draw[->, thick, color=mygreen] (axis cs: 1, {\cOver + #1 * \stepssmall}) -- (axis cs: 1 - \arrowLengthssmall, {\cOver + #1 * \stepssmall});
    					\draw[->, color=mygreen, thick] (axis cs: {1 - #1 * \stepssmall}, \cOver) -- (axis cs: {1- #1 * \stepssmall}, \cOver + \arrowLengthssmall);
    					\draw[->, color=mygreen, thick] (axis cs: {1 - #1 * \stepssmall}, 1) -- (axis cs: {1- #1 * \stepssmall}, {1 - \arrowLengthssmall});
    				}
    				\pgfplotsinvokeforeach{1,...,\numArrowssmallprime}{
    					\draw[->, thick, color=mygreen] (axis cs: -1, {-1 + #1 * \stepssmallprime}) -- (axis cs: {-1 + \arrowLengthssmallprime}, {-1 + #1 * \stepssmallprime});
    					\draw[->, thick, color=mygreen] (axis cs: \cUnder, {-1 + #1 * \stepssmallprime}) -- (axis cs: \cUnder - \arrowLengthssmallprime, {-1 + #1 * \stepssmallprime});
    					\draw[->, thick, color=mygreen] (axis cs: {\cUnder - #1 * \stepssmallprime}, -1) -- (axis cs: {\cUnder - #1 * \stepssmallprime}, -1 + \arrowLengthssmallprime);
    					\draw[->, thick, color=mygreen] (axis cs: {\cUnder - #1 * \stepssmallprime}, \cUnder) -- (axis cs: {\cUnder - #1 * \stepssmallprime}, {\cUnder - \arrowLengthssmallprime});
    
    				}
    				\pgfplotsinvokeforeach{1,...,\numArrowsabprime}{
    					\draw[->, color=myred, thick] (axis cs: \bValprime, {\aValprime + #1 * \stepsabprime}) -- (axis cs: \bValprime - \arrowLengthab, {\aValprime + #1 * \stepsabprime});
    					\draw[->, color=myred, thick] (axis cs: \aValprime, {\aValprime + #1 * \stepsabprime}) -- (axis cs: {\aValprime + \arrowLengthab}, {\aValprime + #1 * \stepsabprime});
    					\draw[->, color=myred, thick] (axis cs: {\bValprime - #1 * \stepsabprime}, \aValprime) -- (axis cs: {\bValprime - #1 * \stepsabprime}, \aValprime + \arrowLengthab);
    					\draw[->, color=myred, thick] (axis cs: {\bValprime - #1 * \stepsabprime}, \bValprime) -- (axis cs: {\bValprime - #1 * \stepsabprime}, {\bValprime - \arrowLengthab});
    				}
    			\end{scope}

    			\node[circle, draw, fill=black, inner sep=1.5pt] at (axis cs: -0.7, -0.7) {};
    			\node[font=\small, anchor=south east, color=black] at (axis cs: -0.7, -0.7) {$c_1 \bm{1}$};
    
    			\node[circle, draw, fill=black, inner sep=1.5pt] at (axis cs: 0.3, 0.3) {};
    			\node[font=\small, anchor=north west, color=black] at (axis cs: 0.3, 0.3) {$c_3 \bm{1}$};
    
    			\node[circle, draw, fill=black, inner sep=1.5pt] at (axis cs: 0.85, 0.85) {};
    			\node[font=\small, anchor=north west, color=black] at (axis cs: 0.85, 0.85) {$c_5 \bm{1}$};
    		\end{axis}
    	\end{tikzpicture}
    	\caption{Induced invariant hypercubes of the signal. The hatched green areas correspond to different invariant hypercubes defined by \emph{successive} fixed points of $s$, while the blue and the red areas correspond to the invariant hypercubes defined by intervals $[a,b]$ and $[a',b']$, respectively.}
    	\label{subfig:invariant_hypercubes_for_signal_function_example}
    \end{subfigure}
    \caption{Illustration of Proposition~\ref{prop:invariance_of_hypercubes_defined_by_fixed_points}.} \vspace{-.5cm}
    \label{fig:invariance_of_hypercubes_defined_by_fixed_points}
\end{figure}

Our analysis reveals two key insights about the system behavior. First, Proposition~\ref{prop:monotone_flow} establishes that almost all trajectories converge to equilibrium states. Second, Proposition~\ref{prop:invariance_of_hypercubes_defined_by_fixed_points} shows that trajectories remain bounded within hypercubes defined by the signal function's fixed points. More precisely, when initial conditions start within a hypercube $[a, b]^N$ satisfying the conditions of Proposition~\ref{prop:invariance_of_hypercubes_defined_by_fixed_points}-\ref{prop:invariance_of_hypercubes_defined_by_fixed_points:item:2}, both the trajectory and its eventual equilibrium state must remain confined within that same hypercube. This property helps us predict and localize where the system will ultimately settle. The following section is dedicated to the characterization of the equilibria of the dynamics \eqref{eq:dynamic}.

\section{Characterization of Equilibria: From Fully Synchronized to Clustered States} \label{sec:equilibria}

Having established that trajectories converge to the set of equilibria of \eqref{eq:dynamic} (Propositions~\ref{prop:monotone_flow} and Lemma~\ref{lemma:SOP_flow}), we now characterize the structure of these points. An equilibrium state $\bm{x}^* \in \mathcal{X}$ satisfies $\dot{\bm{x}}^* = \bm{0}$, which from \eqref{eq:dynamic} implies:
\begin{equation} \label{eq:equilibrium_condition}
	\bm{x}^* = \bm{D}^{-1} \bm{A} \bm{s}(\bm{x}^*).
\end{equation}

Component-wise, this fundamental condition states that for every agent $i \in \Vcal$, its equilibrium state $x_i^*$ must equal the weighted average of the signals received from its neighbors:
\begin{equation}\label{eq:equilibrium_condition_scalar}
	x_i^* = \frac{1}{d_i} \sum_{j \in \mathcal{N}_i} s(x_j^*).
\end{equation}

Equilibria fall into one of two classes:
\begin{itemize}
    \item \emph{Fully Synchronized Equilibria (FSE):} States of complete agreement, belonging to the FSM ($\bm{x}^* \in \mathcal{S}$), where all agents share a common value $c \in [-1,1]$. i.e., $\bm{x}^*=c\one$.
    \item \emph{Non-Fully Synchronized Equilibria (NFSE):} States of persistent disagreement, outside of the FSM ($\bm{x}^* \in \mathcal{X} \setminus \mathcal{S}$), where agents' states settle at different values, often forming clusters.
\end{itemize}

\subsection{Properties of Fully Synchronized Equilibria}
For an equilibrium $\bm{x}^* = c\one$ to exist, it must satisfy the equilibrium condition \eqref{eq:equilibrium_condition}:
\begin{equation*}
    c\one = \bm{D}^{-1}\bm{A}\bm{s}(c\one) = s(c)\bm{D}^{-1}\bm{A}\one.
\end{equation*}
Since the matrix $\bm{D}^{-1}\bm{A}$ is row-stochastic, its rows sum to one, meaning $\bm{D}^{-1}\bm{A}\one = \one$. The condition thus simplifies to $c\one = s(c)\one$, which requires $c = s(c)$. This proves a crucial connection:
\begin{lemma}\label{lem:fse_are_fixed_points}
    The set of Fully Synchronized Equilibria (FSE) of the system \eqref{eq:dynamic} is precisely the set of states $\bm{x}^* = c\one$ where $c$ is a fixed point of the signal function $s$, i.e., $c \in \mathrm{Fix}(s)$.
\end{lemma}

Since Lemma~\ref{lemma:existence_of_stable_fixed_point} guarantees that at least one such fixed point $c$ always exists, the existence of at least one FSE is guaranteed. The central question then becomes one of stability. The following theorem establishes an elegant result: the stability of an $N$-dimensional FSE is completely determined by the one-dimensional stability of its corresponding fixed point, as defined in Definition~\ref{def:fixed_point_stability}.

Beyond simply proving existence, our goal is to characterize the equilibria structure more precisely. The local stability of FSE was studied in \cite{couthuresGlobalSynchronizationMultiagent2025}. We recall the main result:

\begin{theorem}[Stability of FSE]\label{thm:local_stability_of_synchronization_equilibria}
		Let $\bm{x}^*=c \one$ be a Fully Synchronized Equilibrium. Then:
		\begin{enumerate}
			\item $\bm{x}^*$ is locally stable if and only if $c$ is a stable fixed point of $s$ ($c \in \mathrm{Fix}^{\bullet}(s)$).
			\item $\bm{x}^*$ is locally asymptotically stable if and only if $c$ is a stable and isolated fixed point of $s$.
			\item $\bm{x}^*$ is unstable if and only if $c$ is not a stable fixed point of $s$ ($c \in \mathrm{Fix}^{\circ}(s)$).
		\end{enumerate}
\end{theorem}
\begin{proof}
    See \cite[Theorem~1]{couthuresGlobalSynchronizationMultiagent2025}.
\end{proof}

This theorem provides a complete characterization of the stability of all possible states of full agreement.

\subsection{Basins of Attraction for Synchronized States}\label{subsec:attraction_basin_of_FSE}

Theorem~\ref{thm:local_stability_of_synchronization_equilibria} tells us which FSE are stable, but not which initial conditions converge to them. By leveraging the invariant hypercubes identified in Proposition~\ref{prop:invariance_of_hypercubes_defined_by_fixed_points}, we can provide a geometric characterization of the basins of attraction for these equilibria. The following proposition shows that these basins are defined by the semi-stable or unstable fixed points of $s$, which act as \emph{separatrixes} \cite{chiangStabilityRegionsNonlinear2015d}. 

\begin{proposition}[Basin of Attraction of FSE]\label{prop:basin_of_attraction_N_dim}

    Let $c^*$ be a fixed point of $s$, defining the FSE $\bm{x}^* = c^*\one$. Let its basin of attraction be $\mathcal{B}(\bm{x}^*):= \{ \bm{x}(0) \in \Xcal \mid \lim_{t\to \infty} \bm{x}(t) = \bm{x}^*\}$. Define the nearest non-stable fixed points bounding $c^*$ as:
    \begin{align*}
        \underline{c} &:= \max (\{ c \in \mathrm{Fix}^{\circ}(s) \mid c < c^* \} \cup \{-1\}), \\
        \overline{c} &:= \min (\{ c \in \mathrm{Fix}^{\circ}(s) \mid c > c^* \} \cup \{1\}).
    \end{align*}

    Then the following equivalences hold:
    
    \begin{enumerate}
        \item The region $\left[ \underline{c}, c^*\right]^N\setminus \{ \underline{c}\one\} $ is a subset of $ \mathcal{B}(\bm{x}^*)$ if and only if $s(x)  > x$, for all $x \in \left( \underline{c}, c^*\right)$, i.e., $\underline{c} \in \mathrm{Fix_R^{\circ}}(s) \cup \{-1\}$.\label{prop:basin_of_attraction_N_dim:item:1}
        \item The region $\left[ c^*, \overline{c}\right]^N \setminus \{ \overline{c}\one\}$ is a subset of $ \mathcal{B}(\bm{x}^*)$ if and only if $s(x) < x$, for all $x \in \left( c^*, \overline{c}\right)$, i.e., $\overline{c} \in \mathrm{Fix_L^{\circ}}(s) \cup \{1\}$.\label{prop:basin_of_attraction_N_dim:item:2}
        \item The region $[\underline{c}, \overline{c}]^N \setminus \{ \underline{c}\one, \overline{c}\one\} \subseteq \mathcal{B}(\bm{x}^*)$ if and only if $\bm{x}^*$ is locally asymptotically stable. i.e., $c^* \in \mathrm{Fix}^{\bullet}(s)$.\label{prop:basin_of_attraction_N_dim:item:3}
    \end{enumerate}
    Those results extend to intervals over $\Scal$ of fixed point $I\subseteq \mathrm{Fix}(x)$.
\end{proposition}

\begin{proof}
    Note that $\underline{c} \one$ and $\overline{c} \one $ are also equilibria of \eqref{eq:dynamic}. Thus, $\underline{c} \one \notin \mathcal{B}(x^*)$ and $\overline{c} \one \notin \mathcal{B}(x^*)$.

    \ref{prop:basin_of_attraction_N_dim:item:1} ($\Rightarrow$) Let us restrict to the FS case since $\Scal$ is forward invariant. Over, $\Scal$ we have a 1-dimensional dynamics: $\dot{x} = s(x) -x := f(x)$ defined over $[-1,1]$.
    
    Since $[\underline{c}, c^*]^N\setminus \{ \underline{c}\one\} \subseteq \mathcal{B}(x^*)$, one has $[\underline{c}, c^*]^N \cap \Scal \setminus \{ \underline{c}\one\}\subseteq \mathcal{B}(x^*)$, then for all $x(0) \in (\underline{c}, c^*]$, we must have $\lim_{t\to\infty} x(t) = x^*$. Moreover, since for all $x \in (\underline{c}, c^*)$, we have $x < c^*$, the function $t \mapsto x(t)$ must be increasing. Therefore, $\dot{x}(t) >0$ for all $t \geq 0$. Then,
	\begin{equation*}
		\dot{x}(t) = s(x(t)) - x(t) > 0 \Leftrightarrow s(x(t)) > x(t).
	\end{equation*}
	This yields $s(x) > x$ for all $x \in (\underline{c}, c^*)$.

	($\Leftarrow$) Let us consider the function $V(\bm{x}) = \frac{1}{2}(\bm{x} - \bm{x}^*)^\top \bm{D}(\bm{x} - \bm{x}^*)  \geq 0$. The time derivative of $V$ along the trajectories of the dynamics is given by
	\begin{align*}
		\dot{V}(\bm{x}) &= (\bm{x} - \bm{x}^*)^\top \!\bm{D} \dot{\bm{x}} = (\bm{x} - \bm{x}^*)^\top\! \bm{A}\bm{s}(\bm{x}) - (\bm{x} - \bm{x}^*)^\top\! \bm{D}\bm{x}\\
        & = (\bm{x} - c^* \one)^\top \!\bm{D}(\bm{s}(\bm{x}) - \bm{x}) - (\bm{x} - c^* \one)^\top \! \bm{L}\bm{s}(\bm{x}),
	\end{align*}
    using the definition of the Laplacian matrix $\bm{L} = \bm{D} - \bm{A}$. Since $\one^\top$ is a left eigenvector of $\bm{L}$ associated with eigenvalue $0$, this rewrites in summation form:
    \begin{align*}
        \dot{V}(\bm{x}) &= \sum_{i=1}^N d_i (x_i - c^*)(s(x_i) - x_i) \\
        &\qquad - \frac{1}{2}\sum_{i,j=1}^N a_{ij} (x_i - x_j)(s(x_i) - s(x_j)).
    \end{align*}
    One has $\{\bm{x} \in \Xcal \mid \dot{V}(\bm{x}) = 0\} \cap (\left[ \underline{c}, c^*\right]^N\setminus \{ \underline{c}\one\}) = \{\bm{x}^*\}$, since $\bm{L}\one = 0$ and $c^*$ is the only fixed point of $s$ in $\left( \underline{c}, c^*\right]$ (since $s(x)  > x$, for all $x \in \left( \underline{c}, c^*\right)$). Additionally, over the hypercube $[\underline{c}, c^*)^N\setminus \{ \underline{c}\one\}$, the first term of $\dot{V}(\bm{x})$ is negative since there exists at least one $i \in \Vcal$ such that $x_i\neq c_{\mathrm{L}}^*$ and then verify $(x_i - c^*)(s(x_i) - x_i) < 0$ by the assumption $s(x_i) > x_i$ and $\sum_{i,j=1}^N a_{ij} (x_i - x_j)(s(x_i) - s(x_j)) \geq 0$ by non-decreasing-ness of $s$. Therefore, $V$ is a Lyapunov function for the dynamics \eqref{eq:dynamic} over $[\underline{c}, c^*]^N\setminus \{ \underline{c}\one\}$.
    
    Let $n_0 \in \N$ such that $[\underline{c} + 1/n_0, c^*] \subset (\underline{c}, c^*]$. For any $i \in \Vcal$, let us note $H_n^i = [\underline{c}, c^*]^{i-1} \times[\underline{c} + 1/n, c^*] \times [\underline{c}, c^*]^{N-i} $ for $n \geq n_0$, this hypercube is a compact forward invariant for the dynamics \eqref{eq:dynamic} from Corollary~\ref{cor:invariance_of_hypercubes_defined_by_fixed_points}-\ref{prop:invariance_of_hypercubes_defined_by_fixed_points:item:2bis}. Additionally, one has $\dot{V}(\bm{x}) < 0$ for any $\bm{x} \in H_n^i \setminus \{c^* \one\}$. Thus, by LaSalle's invariance principle \cite[Theorem~4.4]{khalil_nonlinear_2002}, one has $H_n^i \subseteq \mathcal{B}(\bm{x}^*)$. Then, since this is true for all $i \in \Vcal$ and $n \geq n_0$, we have that every $H_n^i \subseteq \mathcal{B}(x^*)$ and thus
	\begin{equation*}
		\bigcup_{n \geq n_0}\bigcup_{i\in\Vcal} H_n^i \subseteq \mathcal{B}(x^*) \Leftrightarrow [\underline{c}, x^*]^N \setminus\{c_{\mathrm{L}^*}\one \} \subseteq \mathcal{B}(x^*).
	\end{equation*}

	\ref{prop:basin_of_attraction_N_dim:item:2} The proof follows by symmetry from the previous.
	
	\ref{prop:basin_of_attraction_N_dim:item:3} This is a corollary of \ref{prop:basin_of_attraction_N_dim:item:1} and \ref{prop:basin_of_attraction_N_dim:item:2} or \cite[Proposition~4]{couthuresGlobalSynchronizationMultiagent2025}.

    The extension to intervals over $\Scal$ uses the same argument except that in that case $\{\bm{x} \in \Xcal \mid \dot{V}(\bm{x}) = 0\} = I^N \cap \Scal$ and is not reduced to a singleton. The results follow from LaSalle's invariance principle \cite[Theorem~4.4]{khalil_nonlinear_2002}.
\end{proof}

\begin{figure}[t]
	\centering
	\begin{subfigure}[t]{0.48\linewidth}
	\centering
	\begin{tikzpicture}[scale=0.65,
		vec/.style={
              decorate,
              decoration={
                markings,      
                mark=at position 0.25 with {\arrow{stealth}},
                mark=at position 0.85 with {\arrow{stealth}}
              }}]
		\begin{axis}[
			axis lines=middle,
			xlabel={$x$},
			ylabel={},
			xmin=-1.1, xmax=1.1,
			ymin=-1.1, ymax=1.1,
			xtick={-1,  -0.5,  0, 0.5, 1},
			ytick={-1, -0.5, 0, 0.5, 1},
			xticklabel style={font=\small, anchor=north},
			yticklabel style={font=\small, anchor=north west},
			legend pos=north west,
			legend cell align={left},
			grid=major,
			width=8cm,
			height=8cm,
		]

		\addplot[domain=-1:1, samples=2, dashed, color=gray, thick] {x};
		\addlegendentry{\small $y=x$}

		\addplot[domain=-1:1, samples=501, color=myblue, very thick] {tanh(2.5*x)};
		\addlegendentry{\small $s(x)$}
		
		\addlegendimage{only marks, mark=o, fill=white, draw=black}
		\addlegendentry{\small $\mathrm{Fix}^{\circ}(s)$}
		
		\addlegendimage{only marks, mark=*, fill=black, draw=black}
		\addlegendentry{\small $\mathrm{Fix}^{\bullet}(s)$}

		\node[circle, draw, fill=black, inner sep=1.5pt] at (axis cs: -0.98562, -0.98562) {};
		\node[circle, draw, fill=black, inner sep=1.5pt] at (axis cs: 0.98562, 0.98562) {};
		\node[circle, draw, fill=white, inner sep=1.5pt] at (axis cs: 0, 0) {};

		\end{axis}

		\begin{axis}[
			axis y line=none,
			axis x line*=bottom,
			xmin=-1.1, xmax=1.1,
			ymin=0, ymax=1, 
			xtick={-1, -0.5, 0, 0.5, 1},
			xticklabel style={font=\small, anchor=north},
			xlabel={$x$},
			width=8cm,
			height=2.5cm,
			yshift=-1cm, 
			axis line style={-}, 
		]
		
		\begin{pgfonlayer}{background}
			\fill[myblue!30, opacity=0.8] (axis cs: -1, 0) rectangle (axis cs: 0, 0.9);
			\fill[mygreen!30, opacity=0.8] (axis cs: 0, 0) rectangle (axis cs: 1, 0.9);
		\end{pgfonlayer}

		\draw[postaction=vec, color=black,  thick] (axis cs: -1, 0.3) -- (axis cs: -0.98562, 0.3);
		\draw[postaction=vec, color=black,  thick] (axis cs: 0, 0.3) -- (axis cs: 0.98562, 0.3);
		\draw[postaction=vec, color=black,  thick] (axis cs: 0, 0.3) -- (axis cs: -0.98562, 0.3);
		\draw[postaction=vec, color=black,  thick] (axis cs: 1, 0.3) -- (axis cs: 0.98562, 0.3);

		\node[circle, draw, fill=black, inner sep=2pt] at (axis cs: -0.98562, 0.3) {};
		\node[circle, draw, fill=black, inner sep=2pt] at (axis cs: 0.98562, 0.3) {};
		\node[circle, draw, fill=white, inner sep=2pt] at (axis cs: 0, 0.3) {};

		\node[text width=2cm, align=center, font=\small] at (axis cs: -0.5, 0.6) {$\mathcal{B}(-0.95)$};
		\node[text width=2cm, align=center, font=\small] at (axis cs: 0.5, 0.6) {$\mathcal{B}(0.95)$};
		
		\end{axis}

	\end{tikzpicture}
	\caption{$s(x) = \tanh(2.5x)$.}
	\label{fig:fs_dynamics_examples:a}
\end{subfigure}
\hfil\begin{subfigure}[t]{0.48\linewidth}
	\centering
	\begin{tikzpicture}[scale=0.65,
		vec/.style={
              decorate,
              decoration={
                markings,      
                mark=at position 0.25 with {\arrow{stealth}},
                mark=at position 0.85 with {\arrow{stealth}}
              }}]
		\begin{axis}[
			axis lines=middle,
			xlabel={$x$},
			ylabel={},
			xmin=-1.1, xmax=1.1,
			ymin=-1.1, ymax=1.1,
			xtick={-1,  -0.5,  0, 0.5, 1},
			ytick={-1, -0.5, 0, 0.5, 1},
			xticklabel style={font=\small, anchor=north},
			yticklabel style={font=\small, anchor=north west},
			legend pos=north west,
			legend cell align={left},
			grid=major,
			width=8cm,
			height=8cm,
		]

		\addplot[domain=-1:1, samples=2, dashed, color=gray, thick] {x};
		\addlegendentry{\small $y=x$}

		\addplot[domain=-1:1, samples=501, color=myblue, very thick] {x - min(sin(deg(2 * pi *x))/( 2 * pi) , sin(deg(2 * pi*x+pi))/( 2 * pi))};
		\addlegendentry{\small $s(x)$}

		\addlegendimage{only marks, mark=o, fill=white, draw=black}
		\addlegendentry{\small $\mathrm{Fix}^{\circ}(s)$}
		
		\addlegendimage{only marks, mark=*, fill=black, draw=black}
		\addlegendentry{\small $\mathrm{Fix}^{\bullet}(s)$}

		\pgfplotsinvokeforeach{-2,...,1}{
			\node[circle, draw, fill=white, inner sep=1.5pt] at (axis cs: #1/2, #1/2) {};
		}
		\node[circle, draw, fill=black, inner sep=1.5pt] at (axis cs: 1, 1) {};

		\end{axis}

		\begin{axis}[
			axis y line=none,
			axis x line*=bottom,
			xmin=-1.1, xmax=1.1,
			ymin=0, ymax=1, 
			xtick={-1, -0.5, 0, 0.5, 1},
			xticklabel style={font=\small, anchor=north},
			xlabel={$x$},
			width=8cm,
			height=2.5cm,
			yshift=-1cm, 
			axis line style={-}, 
		]
		
		\begin{pgfonlayer}{background}
			\fill[myblue!30, opacity=0.8] (axis cs: -1, 0) rectangle (axis cs: -0.5, 0.9);
			\fill[mygreen!30, opacity=0.8] (axis cs: -0.5, 0) rectangle (axis cs: 0, 0.9);
			\fill[myred!30, opacity=0.8]  (axis cs: 0, 0) rectangle (axis cs: 0.5, 0.9);
			\fill[myorange!30, opacity=0.8](axis cs: 0.5, 0) rectangle (axis cs: 1, 0.9);
		\end{pgfonlayer}

		\pgfplotsinvokeforeach{-2,...,1}{
			\draw[postaction=vec, color=black,  thick] (axis cs: {#1/2 + 0.02}, 0.3) -- (axis cs: {((#1+1)/2) - 0.02}, 0.3);
		}
		
		\pgfplotsinvokeforeach{-2,...,2}{
			\node[circle, draw, fill=white, inner sep=2pt] at (axis cs: #1/2, 0.3) {};
		}
		\node[circle, draw, fill=black, inner sep=2pt] at (axis cs: 1, 0.3) {};

		\node[text width=2cm, align=center, font=\small] at (axis cs: -0.75, 0.6) {$\mathcal{B}(-0.5)$};
		\node[text width=2cm, align=center, font=\small] at (axis cs: -0.25, 0.6) {$\mathcal{B}(0)$};
		\node[text width=2cm, align=center, font=\small] at (axis cs: 0.25, 0.6) {$\mathcal{B}(0.5)$};
		\node[text width=2cm, align=center, font=\small] at (axis cs: 0.75, 0.6) {$\mathcal{B}(1)$};
		
		\end{axis}

	\end{tikzpicture}
	\caption{$s(x) = x - \min(\sin(\alpha x),\allowbreak \sin(\alpha x + \pi)) / \alpha$, with $\alpha = 2\pi$.}
	\label{fig:fs_dynamics_examples:b}
\end{subfigure}
\caption{Illustration of the attraction basins for dynamics restricted to FSM (i.e., $\dot{x}\one = (s(x) - x) \one$) for two different signal functions $s(x)$. The first plots show the function $s(x)$ and the identity line $y=x$ and exhibits the fixed points of $s(x)$. The stable fixed points are represented by a black circle and the unstable ones by a white circle. The second plots show the basins of attraction of the different equilibria over $\Scal$.} \vspace{-.5cm}
\label{fig:fs_dynamics_examples}
\end{figure}
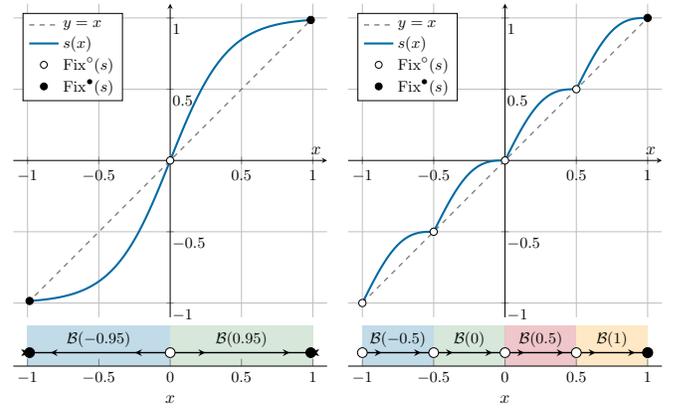

Proposition~\ref{prop:basin_of_attraction_N_dim} yields geometric insights into the system's global behavior. Since the system is cooperative and its trajectories are ordered, the one-dimensional dynamics on the FSM effectively governs the behavior within the entire $N$-dimensional hypercubes defined by successive fixed points. This implies that any trajectory starting within, or later entering, an invariant hypercube such as $(\underline{c}, \overline{c})^N$ is guaranteed to converge asymptotically to a stable equilibrium on the FSM.

This principle of convergence has a big impact on the potential location of any NFSE. If an equilibrium exists outside the FSM, it cannot lie within any of the invariant ``trapping regions" that guarantee convergence to a synchronized state (the green areas in Figure~\ref{fig:invariance_of_hypercubes_defined_by_fixed_points}). Therefore, a necessary condition for the existence of NFSE is that they must lie in the complement between these regions with respect to $\mathcal{X}\setminus [c_i, c_{i+1}]^N$ for $i\in \{1, \dots,M\}$, where $c_1 < c_2< \dots < c_M$ are the fixed points of $s$ ordered in increasing order (such as the complement of the blue area with respect to the green ones). This geometric constraint provides the foundation for our investigation into the conditions under which such states of persistent disagreement can emerge.

\subsection{Emergence of Non-Fully Synchronized Equilibria}

The analysis so far has focused on equilibria characterized by full synchronization. We now investigate the conditions under which states of persistent disagreement, or NFSE, can emerge. This also raises a fundamental question: can NFSE act as stable attractors for the dynamics? The following example provides a constructive proof of their existence and offers intuition for the mechanisms that enable them.

\begin{example}[Emergence of a Stable NFSE on a Line Graph]\label{example:non_sync_stability_N5}
	To provide a constructive proof that \emph{asymptotically stable} NFSE can exist, we consider a line graph $\mathcal{G}$ with $N=5$ agents and an odd sigmoidal signal function, $s(-x) = -s(x)$. This system possesses an anti-symmetric invariant subspace,
	\begin{equation*}
		\mathcal{M} = \{ \bm{x} \in \mathcal{X} \mid x_3=0, x_2 = -x_4, \text{ and } x_1 = -x_5 \}.
	\end{equation*} 
	Then, the set $\mathcal{M}$ is an invariant manifold under the dynamics \eqref{eq:dynamic}. Indeed, for any $\bm{x} \in \mathcal{M}$:
	\begin{align*}
		\dot{x}_1 &= s(x_2) - x_1 = -s(-x_2) + x_5 = -s(x_4) + x_5 = -\dot{x}_5\\
		\dot{x}_2 &= \frac{s(x_1)+s(x_3)}{2} - x_2 = -\left( \frac{s(-x_1)}{2} + x_4 \right) = -\dot{x}_4 \\
		\dot{x}_3 &= \frac{s(x_2)+s(x_4)}{2} - x_3 = \frac{s(x_2)+s(-x_2)}{2} - 0 = 0.
	\end{align*}
	Therefore, any trajectory starting on $\mathcal{M}$ remains on $\mathcal{M}$. The dynamics on this manifold reduces to the two-dimensional system in variables $(\tilde{x}_1, \tilde{x}_2) = (x_1, x_2)$ given by
	\begin{equation}\label{eq:dynamics_on_line_N5}
		\begin{cases}
			\dot{\tilde{x}}_1 =  s(\tilde{x}_2) - \tilde{x}_1, \\
			\dot{\tilde{x}}_2 =  \frac{1}{2}s(\tilde{x}_1) - \tilde{x}_2.
		\end{cases}
	\end{equation}
	
    The asymptotic stability of equilibria in this system depends critically on the gain at the origin, $K=s'(0)$. A local bifurcation analysis, visualized comprehensively in Figure~\ref{fig:bifurcation_diagram_N5}, reveals a two-stage process. First, at $K_{\mathrm{bif}} = \sqrt{2}$, the FSE becomes unstable on $\mathcal{M}$, giving rise to two symmetric branches of NFSE. However, these equilibria remain unstable in the full state space $\mathcal{X}$ due to dynamics transverse to the manifold. Second, at a higher critical gain $K_{\mathrm{stab}}$, a transversal stabilization bifurcation occurs, rendering the NFSE locally stable in $\mathcal{X}$.
\end{example}

\begin{figure}[t]
	\centering
	\includegraphics[width=\linewidth]{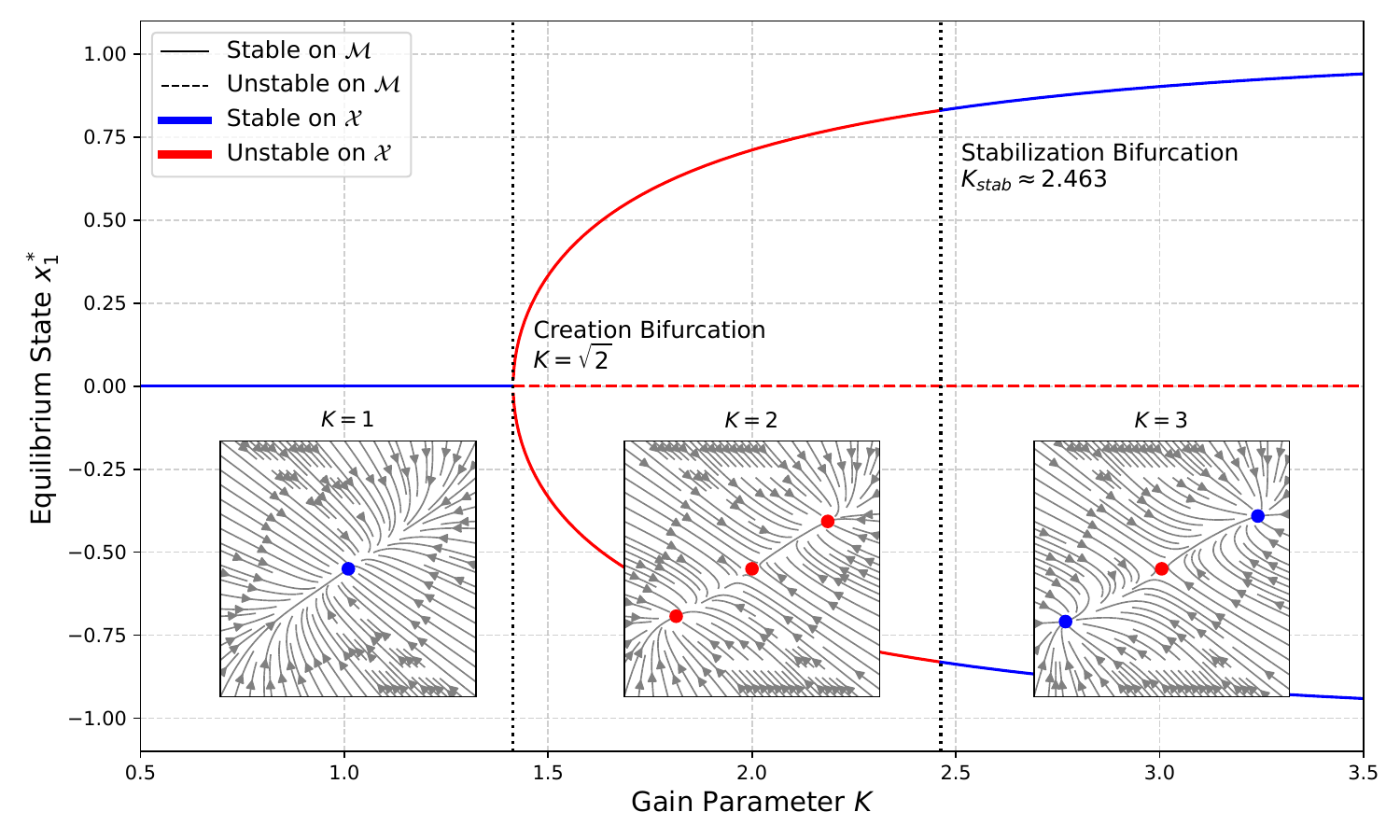}
    \caption{Bifurcation diagram for the $N=5$ line graph with $s(x)=\tanh(Kx)$ for $K \in [0.5,3.5]$, illustrating the distinction between stability on the invariant manifold $\mathcal{M}$ and stability in the full state space $\mathcal{X}$. The plot shows the equilibrium state $x_1^*$ as a function of the gain parameter $K$.
	\emph{Stability on $\mathcal{M}$ (lines):} The FSE at $x_1^*=0$ loses stability at $K=\sqrt{2}$ (dashed line), creating two branches of NFSE that are stable on the manifold (solid lines).
	\emph{Stability in $\mathcal{X}$ (red/blue colors):} The FSE becomes unstable in the full space at $K=\sqrt{2}$ (red dashed line). Crucially, the new NFSE branches are also born unstable in $\mathcal{X}$ (solid red lines), despite their stability on $\mathcal{M}$. They only become stable in the full space (solid blue lines) after the transversal stabilization bifurcation at $K_{\mathrm{stab}} \approx 2.463$.
    The inset phase portraits visualize the dynamics \eqref{eq:dynamics_on_line_N5} on $\mathcal{M}$ at three representative values of $K$.}\vspace{-0.3cm}
	\label{fig:bifurcation_diagram_N5}
\end{figure}

This example is not an isolated curiosity; it reveals the fundamental conditions necessary for any NFSE to exist. First, the loss of stability of the FSE at the origin was the critical event that allowed a new NFSE to emerge. This suggests a general principle: the existence of at least one non-stable fixed point ($c^\circ \in \mathrm{Fix}^\circ(s)$) is a prerequisite for breaking the system's tendency to synchronization. Second, the emergent NFSE was characterized by a clear partitioning of agent states around this unstable point. This hints that any NFSE must be structurally organized around such a ``splitting point." Finally, the shape of the signal function $s(x)$ was the ultimate cause; its steep slope at the origin created the instability. This shows the importance of the local properties of $s(x)$, specifically the existence of points where it generates a repelling rather than a restoring force. The following theorem formalizes and generalizes these important insights.

\subsection{Necessary Conditions for the Existence of NFSE} \label{subsec:non_sync_equilibrium}

\begin{theorem}[Necessary Conditions for NFSE]\label{thm:conditions_for_nfse}
    Let $\bm{x}^*$ be an NFSE of the dynamics \eqref{eq:dynamic}. i.e., $\bm{x}^*$ is an equilibrium that lies outside the FSM, $\bm{x}^* \notin \mathcal{S}$. Then:
	\begin{enumerate}
		\item The set of non-stable fixed points must be non-empty, i.e., $\mathrm{Fix}^{\circ}(s) \neq \emptyset$. \label{thm:conditions_for_nfse:itm:1}
        \item There must exist a non-stable fixed point $c^\circ \in \mathrm{Fix}^{\circ}(s)$ that splits the agents into two non-empty sets, $\mathcal{I}=\{i \mid x_i^* < c^\circ\}$ and $\mathcal{J}=\{j \mid x_j^* > c^\circ\}$. Furthermore, both of these sets must contain at least two agents: $|\mathcal{I}| \geq 2$ and $|\mathcal{J}| \geq 2$.\label{thm:conditions_for_nfse:itm:2}
        \item Let $x_m^* = \min_{i\in \Vcal} x_i^*$ and $x_M^* = \max_{i\in \Vcal} x_i^*$. The signal function $s$ must possess both a left-unstable fixed point, $c_\mathrm{L}^{\circ}$, and a right-unstable fixed point, $c_\mathrm{R}^{\circ}$, that lie strictly between the extremal states of the equilibrium: $x_m^* < c_\mathrm{L}^{\circ} \leq c_\mathrm{R}^{\circ} < x_M^*$.\label{thm:conditions_for_nfse:itm:3}
	\end{enumerate}
\end{theorem}

\begin{proof}
	\ref{thm:conditions_for_nfse:itm:1} We will proceed by contradiction. Assume there exists an NFSE $\bm{x}^* \notin \mathcal{S}$ while $\mathrm{Fix}^{\circ}(s)$ is empty. 
	
	If $\mathrm{Fix}^{\circ}(s)$ is empty, then $s$ neither reaches the identity line from below nor exits it above. By the continuity of $s$, this implies there is either a unique fixed point $c^{\bullet}$ of $s$ or a continuum of them (i.e., there exists an interval $I \subset [-1,1]$ over which $s$ is an exact estimation). The same reasoning applies to both situations by taking $c^{\bullet}$ as the minimal fixed point in the case of a continuum. Then, $x \mapsto s(x-c^{\bullet}) - c^{\bullet}$ is an underestimation. By \cite[Proposition~3]{couthuresGlobalSynchronizationMultiagent2025}, if $s$ is an underestimation over $[-1,1]$, it implies that $\mathcal{S}$ is globally asymptotically stable, thus there is no non-synchronized equilibrium contradicting the premise.

	\ref{thm:conditions_for_nfse:itm:2} Let $\bm{x}^* \notin \mathcal{S}$ be given. Let $x_m^* = \min_{i \in \Vcal} x_i^*$ and $x_M^* = \max_{i \in \Vcal} x_i^*$. We first show by contradiction that there must exist some $c^\circ \in \mathrm{Fix}^{\circ}(s)$ such that $x_m^* < c^\circ < x_M^*$.

	Suppose that for all $c^\circ \in \mathrm{Fix}^{\circ}(s)$, $c^\circ \leq x_i^* \leq 1 $ for all $i\in \Vcal$. This must be true for their maximum $c^\circ_{\max} = \max(\mathrm{Fix}^\circ(s))$, then $c_{\max}^{\circ} \leq x_i \leq 1$. By assumption, $\bm{x}^*\notin \Scal$, thus we must have $x_M^* > c_{\max}^{\circ}$ and $x_m < 1$. i.e., $\bm{x}^* \in  [c^\circ_{\max}, 1]^N \setminus\{c^\circ_{\max} \one, \one \}$. Now there are two possible situations:
    \begin{itemize}[wide]
        \item $\mathrm{Fix}^{\bullet}(s) \cap (c^\circ_{\max} ,1] = \emptyset$. Then, by continuity of $s$ we must have $s(x) \leq x$ for all $x > c_{\max}^{\circ}$. Since $c_{\max}^{\circ}$ is the maximum of $\mathrm{Fix}^{\circ}(s)$, the inequality is strict, otherwise it would exist another unstable fixed point of $s$ greater than $c_{\max}^{\circ}$. i.e., $s(x) < x$ for all $x > c_{\max}^{\circ}$in this case, from Proposition~\ref{prop:basin_of_attraction_N_dim}-\ref{prop:basin_of_attraction_N_dim:item:2}, one has $\bm{x}^* \in \mathcal{B}(c_{\max}^{\circ} \one)$.
        \item $\mathrm{Fix}^{\bullet}(s) \cap (c^\circ_{\max} ,1] \neq \emptyset$. If this set is reduced to a singleton $c^\bullet$ in this case by Proposition~\ref{prop:basin_of_attraction_N_dim}-\ref{prop:basin_of_attraction_N_dim:item:3} one has $[c^\circ_{\max}, 1]^N \setminus \{c_{\max}^{\circ}\one, \one\} \subseteq \mathcal{B}(c^\bullet\one)$ and implies $\bm{x}^* \in \mathcal{B}(c^\bullet\one)$. If this set is an interval $I$, the reasoning is the same, using the extension of Proposition~\ref{prop:basin_of_attraction_N_dim}-\ref{prop:basin_of_attraction_N_dim:item:3} to intervals and $\bm{x}^* \in \mathcal{B}(I^N \cap \Scal)$. i.e., $\bm{x}^*$ is in the attraction basin of a continuum of equilibria.
    \end{itemize}
    In both cases, $\bm{x}^*$ is in the basin of attraction of FSE equilibria. Moreover, $\bm{x}^* \in \mathcal{B}(\bm{x}^*)$ by definition. Since $\bm{x}^*$ cannot be in two different attraction basins, this contradicts the assumption that for all $c^\circ \in \mathrm{Fix}^{\circ}(s)$, $x_i^* \geq c^\circ$ for all $i \in \Vcal$. Thus, there exists $c^\circ \in \mathrm{Fix}^{\circ}(s)$ such that $x_m^* < c^\circ < x_M^*$. Using \eqref{eq:equilibrium_condition_scalar}, one obtains that:
	\begin{equation*}
		x_m^* = \frac{1}{d_m}\sum_{j=1}^N a_{mj}s(x_j^*) < c^\circ.
	\end{equation*}
	Since $x_m^*$ is a convex combination of the family $(s(x_i^*))_{i\in \Ncal_m}$, this implies there must exist some $j \in \mathcal{N}_m$ with $s(x_j^*) < c^\circ$. Recalling that $s$ is non-decreasing one gets $s(x_m) \leq s(x_j) < c^\circ$.  It is clear that $j,m \in\mathcal{I} = \{ i \in \Vcal \mid s(x_i) < c^\circ\}$. Moreover, since $j\in \mathcal{N}_m$ and $m\notin \mathcal{N}_m$ (since there is self loop in the graph $\Gcal$), this implies, $j\neq m$ and we conclude that $\mathcal{I}$ has at least two elements. Similarly, since $x_M^* > c^\circ$, the same reasoning applies to show that $s(x_M) > c^{\circ}$, and that  $\mathcal{J} = \{ j \in \Vcal \mid s(x_j) > c^\circ\}$ has at least two elements.

    \ref{thm:conditions_for_nfse:itm:3} Let $\bm{x}^*$ be an NFSE, so $x_m^* < x_M^*$. Let $g(x) = s(x)-x$. The equilibrium condition \eqref{eq:equilibrium_condition_scalar} yields:
    \begin{equation*}
        x_m^* = \frac{1}{d_m}\sum_{j=1}^N a_{mj}s(x_j^*) \geq \frac{1}{d_m}\sum_{j=1}^N a_{mj}s(x_m^*) = s(x_m^*).
    \end{equation*}
    We then have $x_m^* \geq s(x_m^*)$ implying $g(x_m^*) \le 0$. A symmetric argument yields $g(x_M^*) \ge 0$. If the inequalities are strict, there exist $x'_m, x'_M$ with $x_m^* \le x'_m < x'_M \le x_M^*$ such that $g(x'_m) < 0$ and $g(x'_M) > 0$.

    Otherwise, if $g(x_m^*) = 0$, since $x_j^* \geq x_m^*$ for all $j \in \Ncal_m$ and $s$ is non-decreasing we must have $s(x_j^*) \geq x_m^*$. Moreover, as $x_m^*$ is a convex combination of $(s(x_j^*))_{j\in \Ncal_m}$, we must have $s(x_j^*) = x_M^*$ for all $j \in \Ncal_m$. Since $\bm{x}^*$ is an NFSE on a connected graph, there exists $k \in \Ncal_m$ with $x_k^* > x_m^*$ (this statement must be true for at least one $m^* \in \arg\max_{i\in \Vcal}x_i^*$, otherwise it imply $\bm{x}\in \Scal$ by propagation of equality to all the graph), implying $s$ is constant on $[x_m^*, x_k^*]$. Thus, for any $x \in (x_m^*, x_k^*]$, $g(x)=x_m^*-x < 0$. Symmetrically, if $g(x_M^*) = 0$, we can find $x$ near and below $x_M^*$ where $g(x) > 0$. Hence, there exist $x'_m, x'_M$ with $x_m^* \le x'_m < x'_M \le x_M^*$ such that $g(x'_m) < 0$ and $g(x'_M) > 0$.

    Let $c_\mathrm{L}^{\circ} := \inf\{x \in (x'_m, x'_M) \mid g(x) \ge 0\}$ and $c_\mathrm{R}^{\circ} := \sup\{x \in (x'_m, x'_M) \mid g(x) \le 0\}$. By continuity of $g$, both are fixed points of $s$. i.e., $g(c_\mathrm{L}^{\circ}) = 0$ and $g(c_\mathrm{R}^{\circ}) =0$. The construction implies that for some $\delta > 0$, $g(x)<0$ on $(c_\mathrm{L}^{\circ}-\delta, c_\mathrm{L}^{\circ})$ and $g(x)>0$ on $(c_\mathrm{R}^{\circ}, c_\mathrm{R}^{\circ}+\delta)$, so $c_\mathrm{L}^{\circ} \in \mathrm{Fix_L^\circ}(s)$ and $c_\mathrm{R}^{\circ} \in \mathrm{Fix_R^\circ}(s)$. By definition, the strict inequalities $x_m^* < c_\mathrm{L}^{\circ} < x_M^*$ and $x_m^* < c_\mathrm{R}^{\circ} < x_M^*$ hold.

    Finally, if $c_\mathrm{L}^{\circ} > c_\mathrm{R}^{\circ}$, then for any $x \in (c_\mathrm{R}^{\circ}, c_\mathrm{L}^{\circ})$, the definitions imply both $g(x) > 0$ and $g(x) < 0$, which would be a contradiction. Thus, $c_\mathrm{L}^{\circ} \le c_\mathrm{R}^{\circ}$, which completes the proof.
\end{proof}

Theorem~\ref{thm:conditions_for_nfse} provides a connection between the properties of the scalar function $s$ and the emergence of disagreement in the $N$-dimensional system.
Item \ref{thm:conditions_for_nfse:itm:3} provides the most immediate practical check: for disagreement to be possible, the signal function \emph{must} possess both left- and right-unstable fixed points. This allows one to rule out NFSE for entire classes of functions by visual inspection. For instance, the signal function in Figure~\ref{fig:fs_dynamics_examples:a}, with its unstable fixed point at the origin, satisfies this condition and can therefore support NFSE (as demonstrated in Example~\ref{example:non_sync_stability_N5}). In contrast, the function in Figure~\ref{fig:fs_dynamics_examples:b}, while possessing right-unstable fixed points, has no left-unstable fixed points. Consequently, Theorem~\ref{thm:conditions_for_nfse} guarantees that NFSE \emph{cannot exist} for any system governed by this signal function, regardless of the network topology.

Furthermore, the structural requirement from item \ref{thm:conditions_for_nfse:itm:2}, that any disagreement must involve at least two agents on either side of a splitting point, has a direct consequence on the minimum network size required to support such a state.

\begin{corollary}\label{coro:N_leq_3}
    If the network size $N \leq 3$, then all equilibria of \eqref{eq:dynamic} are FSE.
\end{corollary}
While Corollary~\ref{coro:N_leq_3} highlights a size constraint, the network's topology plays an equally important role, as we now show.

\subsection{On the Critical Role of Network Topology}\label{sec:signal_and_graph_condition}

Theorem~\ref{thm:conditions_for_nfse} established that the existence of NFSE depends on the shape of the signal function $s$. We now show that the network topology plays an equally important role. This section demonstrates that for any given network, there exists a class of signal functions capable of generating disagreement, and that this capability is directly linked to the network's algebraic connectivity.

We begin by ordering the eigenvalues of the row-stochastic matrix $\bm{D}^{-1} \bm{A}$ as $-1 \leq \lambda_1 \leq \dots \leq \lambda_{N-1} < \lambda_N = 1$. Since the graph is connected, we have $\lambda_{N-1} < 1$. The key metric for our analysis is the second-largest eigenvalue, $\lambda_{N-1}$, which governs the convergence speed of linear consensus \cite{MG10}. Its value is related to the graph's algebraic connectivity, $\mu = 1 - \lambda_{N-1}$. A value of $\lambda_{N-1}$ close to $1$ indicates a sparsely connected or string-like (path-like) structure, while a value close to $0$, or negative, indicates a densely connected graph.

It is a known result that for highly connected graphs such as complete or complete bipartite graphs (where $\lambda_{N-1} \leq 0$), NFSE cannot exist for any non-decreasing signal function \cite{couthuresGlobalSynchronizationMultiagent2025}. We can state the following lemma.  

\begin{lemma}[Highly connected graphs imply FSE]\label{lemma:char_equilibria_dense}
    Let the graph $\Gcal$ be such that $\lambda_{N-1} < 0$. Then, any equilibrium $\bm{x}^*$ of the dynamics \eqref{eq:dynamic} must be an FSE.
\end{lemma}

\begin{proof}
    Let $\bm{P} = \bm{D}^{-1}\bm{A}$. An equilibrium state $\bm{x}^*$ satisfies \eqref{eq:equilibrium_condition}. i.e., $\bm{x}^*=\bm{P}\bm{s}(\bm{x}^*)$. The non-decreasing property of the signal function $s$ implies that $\bm{y}^\top\bm{L}\bm{s}(\bm{y}) =\bm{y}^\top(\bm{D}-\bm{A})\bm{s}(\bm{y}) \geq 0$ for any $\bm{y} \in \Xcal$. This is equivalent to the $\bm{D}$-weighted inner product inequality $\langle \bm{y}, (\bm{I}-\bm{P})\bm{s}(\bm{y}) \rangle_{\bm{D}} \ge 0$.

    Let $\bm{x}^*_\perp$ and $\bm{s}(\bm{x}^*)_\perp$ be the projection of $\bm{x}^*$ and $\bm{s}(\bm{x}^*)$ onto the disagreement subspace (in the $\bm{D}$-orthogonal complement of $\Scal$), respectively. At equilibrium, the inequality reduces to this subspace, yielding $\langle \bm{x}^*_\perp, (\bm{I}-\bm{P})\bm{s}(\bm{x}^*)_\perp \rangle_{\bm{D}} \geq 0$. Projecting the equilibrium condition itself gives $\bm{x}^*_\perp = \bm{P} \bm{s}(\bm{x}^*)_\perp$. By hypothesis, all eigenvalues of $\bm{P}$ on this subspace, $\{\lambda_i\}_{i=1}^{N-1}$, are strictly negative, so $\bm{P}$ is invertible here. Substituting $\bm{s}(\bm{x}^*)_\perp = \bm{P}^{-1}\bm{x}^*_\perp$ into the inequality gives:
    \begin{equation*}
		\langle \bm{x}^*_\perp, (\bm{P}^{-1}-\bm{I})\bm{x}^*_\perp \rangle_{\bm{D}} \ge 0.
	\end{equation*}
    The operator $\bm{P}^{-1}-\bm{I}$ is strictly negative definite on the disagreement subspace, as its eigenvalues $1/\lambda_i - 1$ are all less than $-1$. This quadratic form can only be non-negative if $\bm{x}^*_\perp = \bm{0}$, proving that any equilibrium $\bm{x}^*$ must be an FSE.
\end{proof}

This result formally establishes that for this class of highly connected graphs, disagreement is structurally impossible. We therefore focus on the more general case where $\lambda_{N-1} > 0$. The following proposition provides a sufficient condition, linking the maximal nonlinearity (quantified by $K$) and network connectivity (quantified by $\lambda_{N-1}$), for the existence of a signal function that generates NFSE.

\begin{proposition}\label{prop:non_sync_equilibrium}
	Let $K >0$ and $N \geq 4$. For any graph $\Gcal$, such that $\lambda_{N-1} > 0$. The following statements hold:
	\begin{enumerate}
        \item If $ K \lambda_{N-1} = 1$, there exists a signal function $s$ verifying Assumption~\ref{ass:signal}  such that the dynamics \eqref{eq:dynamic} possesses a continuum of NFSE, in $\mathrm{Span}(\bm{v}_{N-1}) \cap \Xcal$, where $\bm{v}_{N-1}$ is the eigenvector associated to $\lambda_{N-1}$.\label{prop:non_sync_equilibrium_1}
		\item If $K \lambda_{N-1} > 1$, there exists a signal function $s$ verifying Assumption~\ref{ass:signal} (more specifically $s$ is $K$-Lipschitz) such that dynamics \eqref{eq:dynamic} possesses NFSE.\label{prop:non_sync_equilibrium_2}
	\end{enumerate}
	Consequently, if $K \lambda_{N-1} \geq 1$, the fully synchronized manifold $\mathcal{S}$ is not globally asymptotically stable for the dynamics \eqref{eq:dynamic} across all admissible signal functions $s$.
\end{proposition}

\begin{proof}
    The proof of part \ref{prop:non_sync_equilibrium_1} is constructive and provided here. The proof of part \ref{prop:non_sync_equilibrium_2} involves constructing a specific odd signal function and using a topological argument on the stability boundaries of its basins of attraction; due to its length, it is deferred to Appendix~\ref{app:proof_prop5}.
    
	\ref{prop:non_sync_equilibrium_1} By assumption, $\lambda_{N-1} >0$. Under Assumption~\ref{ass:graph}, $\lambda_N= 1$ is simple and we have $\lambda_{N-1} < 1$. Moreover, the eigenvector $\bm{v}_{N-1}$ cannot be aligned with $\bm{v}_N = \one$, thus $\bm{v}_{N-1} \in \Xcal \setminus \Scal$. Since $K \lambda_{N-1} =1$ and $\lambda_{N-1} <1$ one has $K>1$.
	Define the signal function $s(x) = \min(1, \max(-1, Kx))$. This function verifies conditions of Assumption~\ref{ass:signal}. Consider states of the form $\bm{x}_\varepsilon = \varepsilon \bm{v}_{N-1}$ for some scalar $0<\varepsilon \leq 1/K$. Since $\bm{x}_\varepsilon\in [-1/K,1/K]^N$ and the function is linear in this hypercube one has $s(x_i) = K x_i$ for all $i\in \Vcal$. Therefore, $\bm{s}(\bm{x}_\varepsilon) = K \bm{x}_\varepsilon = K \varepsilon \bm{v}_{N-1}$. Injecting it into the dynamics yields:
    \begin{align*}
		\dot{\bm{x}}_\varepsilon &=  \bm{D}^{-1}\bm{A} \bm{s} (\bm{x}_\varepsilon) - \bm{x}_\varepsilon = \varepsilon K \bm{D}^{-1}\bm{A}\bm{v}_{N-1} - \varepsilon \bm{v}_{N-1}\\
        &= \varepsilon (K\lambda_{N-1} -1) \bm{v}_{N-1} = \bm{0}.
	\end{align*}
    Thus, $\bm{x}_\varepsilon$ is an NFSE. This being true for all $0<\varepsilon \leq 1/K$, any $\bm{x}_\varepsilon$ for such $\varepsilon$ is an NFSE. Therefore, there exists a continuum of NFSE in $\mathrm{Span}(\bm{v}_{N-1}) \cap \Xcal$.
\end{proof}

\begin{remark}[Generalization of Sufficiency Condition]
	The proof of \ref{prop:non_sync_equilibrium_2} of this proposition, postponed to Appendix~\ref{app:proof_prop5}, establishes that $\lambda_{N-1}K \geq 1$ is a sufficient condition for the existence of a signal function $s$ inducing an NFSE. This result is analogous to the sufficiency result in \cite[Theorem~6]{fontanMultiequilibriaAnalysisClass}, which uses bifurcation theory to prove the emergence of mixed-sign equilibria under the explicit assumption that the eigenvalue $\lambda_{N-1}$ is simple. On the other hand, our proof relies on a topological argument concerning stability boundaries and does not require this simplicity assumption, thus holding for all network topologies such that $\lambda_{N-1} > 0$, including those with repeated eigenvalues.
\end{remark}

Proposition~\ref{prop:non_sync_equilibrium} provides a constructive proof: if the interplay between nonlinearity and topology is sufficiently strong ($K\lambda_{N-1} \geq 1$), then there \emph{exists} a signal function capable of generating disagreement. This should be read in conjunction with the necessary conditions of Theorem~\ref{thm:conditions_for_nfse}: for an NFSE to exist, the signal function must also have the requisite shape (e.g., possess unstable fixed points). An underestimation function, for example, will never generate an NFSE, regardless of the value of $K\lambda_{N-1}$. Figure~\ref{fig:sync_vs_nosync} provides a visual summary of these dynamics, illustrating the starkly different outcomes on various network topologies when the threshold condition $K\lambda_{N-1} < 1$ is satisfied versus when it is violated.

This synthesis leads to the ultimate question for full synchronization: Do these conditions collapse into a single sharp threshold? Specifically, is the condition $K \lambda_{N-1} < 1$ sufficient to \emph{guarantee} full synchronization for \emph{any} admissible signal function, thereby overpowering any specific shape properties? This question is answered in the next section.

\section{A Sharp Threshold for Full Synchronization}\label{sec:Threshold_global_sync}

This section provides the definitive answer to the question posed above. We prove that the condition $K \lambda_{N-1} < 1$ is indeed sufficient to confine all system equilibria to the FSM, regardless of the specific shape of the signal function $s$. This sufficiency result, when combined with the possibility of disagreement established in Proposition~\ref{prop:non_sync_equilibrium}, formally establishes the condition as a sharp threshold. 
We thereby generalize the condition found in \cite{fontanMultiequilibriaAnalysisClass}, demonstrating that a guarantee previously established for a restrictive class of smooth, sigmoidal interactions holds for any admissible interaction law.

\begin{theorem}[Sharp Threshold for Synchronization]\label{thm:only_sync_equilibria}
	Let $K$ be the Lipschitz constant of $s$. If $ K \lambda_{N-1} < 1$, then, all the equilibria of \eqref{eq:dynamic} are FSE.
\end{theorem}

The proof relies on Ostrowski's theorem on matrix congruence, which we state here for completeness.

\begin{lemma}[Ostrowski {\cite[Theorem~4.5.9]{hornMatrixAnalysis2017}}]
	\label{lemma:ostrowski_singular_sas}
	Let $\bm{A}, \bm{S} \in \R^{N \times N}$ with $\bm{A}$ symmetric. Let the eigenvalues of $\bm{A}$, $\bm{S S}^\top$, and $\bm{SAS}^\top$ be arranged in nondecreasing order, denoted by $\lambda_k(\bm{A})$, $\lambda_k(\bm{S S}^\top)$, and $\lambda_k(\bm{SAS}^\top)$ respectively, for $k=1, \ldots, N$. Then, for each $k=1, \ldots, N$, there exists a real number $\theta_k \in [\lambda_1(\bm{S S}^\top), \lambda_N(\bm{S S}^\top)]$ such that $\lambda_k(\bm{SAS}^\top) = \theta_k \lambda_k(\bm{A})$. 
\end{lemma}
\begin{proof}[Proof of Theorem \ref{thm:only_sync_equilibria}]
    The proof proceeds by contraposition: we assume the existence of an NFSE, $\bm{x}^*$, and we show that this leads to the necessary condition $K \lambda_{N-1} \geq 1$.

    An equilibrium point $\bm{x}^*$ of the system \eqref{eq:dynamic} must satisfy the condition \eqref{eq:equilibrium_condition}. i.e., $\bm{x}^* = \bm{D}^{-1} \bm{A} \bm{s}(\bm{x}^*)$. Since $\bm{x}^*$ is an NFSE, by Theorem~\ref{thm:conditions_for_nfse}-\ref{thm:conditions_for_nfse:itm:2}, there exists $c^\circ \in \mathrm{Fix}^{\circ}(s)$ such that the components of $\bm{x}^*$ and $\bm{s}(\bm{x}^*)$ are not all on one side of $c^\circ$. Specifically, the sets $\mathcal{I} = \{i \mid s(x_i^*) < c^\circ\}$ and $\mathcal{J} = \{j \mid s(x_j^*) > c^\circ\}$ are non-empty.

We now center the equilibrium equation around this fixed point $c^\circ$. Let $\bm{v}^* = \bm{x}^* - c^\circ\bm{1}$. Subtracting $c^\circ\bm{1}$ from the equilibrium equation and using $\bm{D}^{-1}\bm{A}\bm{1}=\bm{1}$ yields:
\begin{equation*}
	\bm{v}^* = \bm{D}^{-1} \bm{A} (\bm{s}(\bm{x}^*) - c^\circ\bm{1}).
\end{equation*}
To recover $\bm{v}^*$ in the right hand side, we define a diagonal matrix $\bm{M} = \mathrm{diag}(m_1, \dots, m_N)$, where 
\begin{equation*}
	m_i = \begin{cases}
		\dfrac{s(x_i^*) - c^\circ}{x_i^* - c^\circ} & \text{if } x_i^* \neq c^\circ, \\
		K & \text{if } x_i^* = c^\circ.
	\end{cases}
\end{equation*}
The Lipschitz continuity of $s$ ensures that each diagonal entry satisfies $0 \leq m_i \leq K$. By construction, we have $\bm{s}(\bm{x}^*) - c^\circ \one = \bm{M} \bm{v}^*$, which transforms the equation into:
\begin{equation} \label{eq:intermediate_eigen_problem}
	\bm{v}^* = \bm{D}^{-1} \bm{A} \bm{M} \bm{v}^*.
\end{equation}

To proceed, we define the vector $\bm{y}^* = \bm{M}^{1/2} \bm{v}^*$ and show that it must be non-zero and contain both positive and negative entries. We argue by contradiction. Suppose $\bm{y}^*$ contains no positive entries, i.e., $y_i^* \leq 0$ for all $i \in \Vcal$. This implies that for any index $j \in \mathcal{J}$ (where $v_j^* > 0$), the corresponding entry $m_j$ must be zero. By definition of $m_j$, this would mean $s(x_j^*) = c^\circ$. This contradicts the fact that $\mathcal{J}$ is the set of indices where $s(x_j^*) > c^\circ$. Therefore, our assumption is false, and $\bm{y}^*$ must contain at least one positive entry. A symmetric argument using the set $\mathcal{I}$ shows it must also contain at least one negative entry. Consequently, $\bm{y}^*$ is a non-zero vector with mixed-sign entries.

Left-multiplying \eqref{eq:intermediate_eigen_problem} by $\bm{M}^{1/2}$ yields $\bm{y}^* = \bm{M}^{1/2} \bm{D}^{-1} \bm{A} \bm{M}^{1/2} \bm{y}^*$. This reveals that $1$ is an eigenvalue of the matrix $\bm{\Theta} = \bm{M}^{1/2} \bm{D}^{-1} \bm{A} \bm{M}^{1/2}$, with corresponding eigenvector $\bm{y}^*$.

To analyze the eigenvalues of $\bm{\Theta}$, we relate it to a symmetric matrix. Let $\tilde{\bm{A}} = \bm{D}^{-1/2} \bm{A} \bm{D}^{-1/2}$ be the symmetrically normalized adjacency matrix, which is symmetric and has the same eigenvalues as $\bm{D}^{-1}\bm{A}$ by similarity. The matrix $\bm{\Theta}$ is similar to the symmetric matrix $\tilde{\bm{\Theta}} = \bm{M}^{1/2} \tilde{\bm{A}} \bm{M}^{1/2}$ via the transformation $\tilde{\bm{\Theta}} = (\bm{D}^{1/2} \bm{M}^{-1/2}) \bm{\Theta} (\bm{M}^{1/2} \bm{D}^{-1/2})$. Therefore, $\bm{\Theta}$ and $\tilde{\bm{\Theta}}$ share the same real eigenvalues, which means $1$ is also an eigenvalue of the symmetric matrix $\tilde{\bm{\Theta}}$. The corresponding eigenvector is $\bm{z}^* = \bm{D}^{1/2} \bm{y}^*$. Since $\bm{D}^{1/2}$ has strictly positive diagonal entries, $\bm{z}^*$ inherits the mixed-sign property of $\bm{v}^*$. 

Since $\tilde{\bm{\Theta}}$ is a nonnegative matrix, by Perron-Frobenius theorem (Lemma~\ref{lemma:perron_frobenius}), the eigenvectors of $\tilde{\bm{\Theta}}$ associated with the largest eigenvalue $\lambda_N(\tilde{\bm{\Theta}})$ must be nonnegative. Since $\bm{z}^*$ has mixed-sign, it cannot be the Perron-Frobenius eigenvector corresponding to the largest eigenvalue. Therefore, $\bm{z}^*$  must be associated with $\lambda_i(\tilde{\bm{\Theta}})$ for $i\in \{1,\dots,N\}$ such that $\lambda_i(\tilde{\bm{\Theta}}) \neq \lambda_N(\tilde{\bm{\Theta}})$. This implies that $1 < \lambda_N(\tilde{\bm{\Theta}})$, and therefore $1$ must be one of the lower eigenvalues, so $1 \leq \lambda_{N-1}(\tilde{\bm{\Theta}})$.

Finally, we use Ostrowski's theorem (Lemma~\ref{lemma:ostrowski_singular_sas}) for matrix congruence on $\tilde{\bm{\Theta}} = (\bm{M}^{1/2})^\top \tilde{\bm{A}} (\bm{M}^{1/2})$. This theorem guarantees that $\lambda_k(\tilde{\bm{\Theta}}) = \theta_k \lambda_k(\tilde{\bm{A}})$ for some scalar $\theta_k \in [\min_i(m_i), \max_i(m_i)]$. Since $0 \leq m_i \leq K$ for all $i$, we have $\theta_k \in [0, K]$. Applying this to the $(N-1)$-th eigenvalue, and noting that $\lambda_{N-1}(\tilde{\bm{A}}) = \lambda_{N-1}$, we get:
\begin{equation*}
	1 \leq \lambda_{N-1}(\tilde{\bm{\Theta}}) = \theta_{N-1} \lambda_{N-1}(\tilde{\bm{A}}) \leq K \lambda_{N-1}.
\end{equation*}
The existence of an NFSE has led to the necessary condition $K \lambda_{N-1} \geq 1$. By contraposition, if $K \lambda_{N-1} < 1$, no such equilibria can exist.
\end{proof}

\begin{remark}[Generalization of Necessity Condition]
	This theorem and its proof generalize the necessity condition to obtain NFSE with sigmoidal signals in \cite[Theorem~4]{fontanMultiequilibriaAnalysisClass} in two key aspects. First, our result applies to a much broader class of signal functions, requiring only that $s$ be non-decreasing and Lipschitz-continuous, whereas the model in \cite{fontanMultiequilibriaAnalysisClass} assumes smooth, odd, and sigmoidal nonlinearities. Second, our proof technique correctly establishes the non-strict inequality $K \lambda_{N-1} \geq 1$ as the necessary condition for disagreement, refining the strict inequality derived for their specific model.
\end{remark}

Theorem~\ref{thm:only_sync_equilibria} provides the condition under which the set of equilibria is confined to the FSM. This has immediate consequences for the long-term behavior of trajectories.

\begin{corollary}[Global Attractivity of the FSM]\label{coro:attractivity_of_S}
	Assume that $ K \lambda_{N-1} < 1$. Then, for almost all initial conditions $\bm{x}(0) \in \Xcal$, one has $\lim_{t\to \infty} \bm{x}(t) \in \Scal$. If additionally, $s$ is strictly increasing and smooth (for simplicity), this is verified for all initial conditions in $\Xcal$. That is, the fully synchronized manifold $\mathcal{S}$ is globally attractive.
\end{corollary}

\begin{proof}
	From Proposition~\ref{prop:monotone_flow}, we have that for almost all initial conditions, the associated trajectory converges to an equilibrium. Since the only equilibria are synchronized equilibria, the FSM is globally attractive. If $s$ is strictly increasing and smooth, from Lemma~\ref{lemma:SOP_flow}, we have the previous verified for all initial conditions, leading to the global attractiveness of $\Scal$.  
\end{proof}

\begin{remark}[Attractivity vs. Lyapunov Stability]\label{rem:attract_vs_stable}
	It is important to distinguish the global attractivity of $\mathcal{S}$ guaranteed by Corollary~\ref{coro:attractivity_of_S} from its Lyapunov stability. Even when the threshold is met, it is possible for the trajectories to transiently move away from $\mathcal{S}$ before finally returning, which excludes a general statement on global asymptotic stability.
\end{remark}

\begin{figure}
	\centering
	\includegraphics[width=\linewidth]{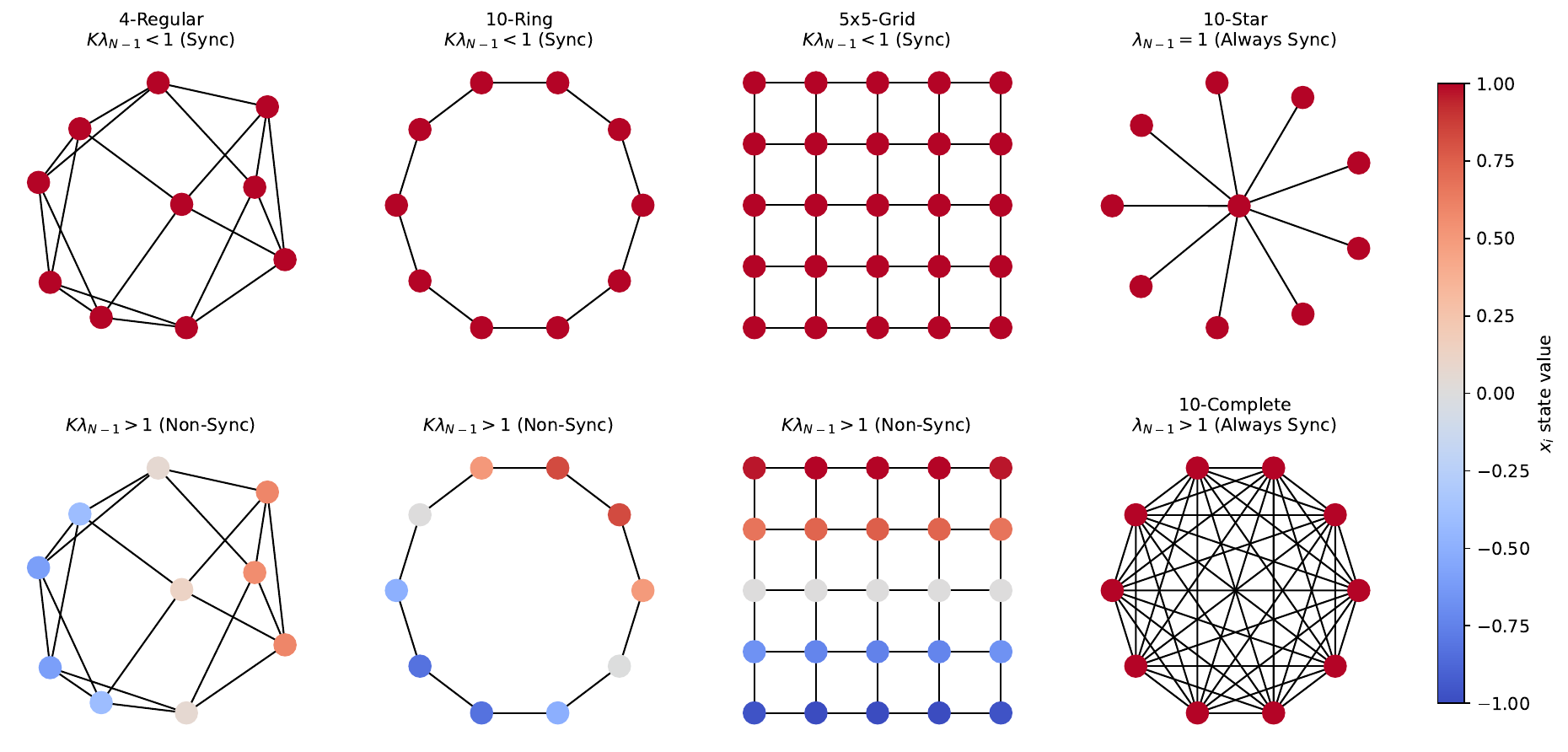}
	\caption{Visual validation of the sharp synchronization threshold on various network topologies. The figure displays the final equilibrium state, with node colors representing the agent state $x_i^* \in [-1,1]$. For diverse topologies (columns 1-3), the top panels confirm that meeting the condition $K\lambda_{N-1} < 1$ guarantees convergence to an FSE, as predicted by Corollary~\ref{coro:attractivity_of_S}. In contrast, the bottom panels demonstrate that violating this threshold leads to an NFSE, characterized by robust clustering. Column 4 highlights the special case of highly connected graphs (e.g., Star and Complete), which are intrinsically robust to disagreement and always synchronize. All simulations use the signal function $s(x) = \max(-1, \min(1, Kx))$, with identical initial conditions for top and bottom panels to ensure a fair comparison.\vspace{-0.5cm}}
	\label{fig:sync_vs_nosync}
\end{figure}

The distinction between attractivity and stability highlighted in Remark~\ref{rem:attract_vs_stable} asks for further discussion. The sharp threshold $K\lambda_{N-1} < 1$ guarantees that no NFSE can exist to act as alternative attractors. It ensures that all trajectories must eventually converge to the FSM. However, it does not constrain the transient behavior. The standard Lyapunov function for consensus, $V(\bm{e}) = \norm{\bm{e}}_{\bm{D}}^2/2$, may not be strictly decreasing everywhere if some eigenvalues of $\bm{D}^{-1}\bm{A}$ are negative and large in magnitude, potentially allowing for transient growth of the disagreement vector $\bm{e}$. 

To guarantee not just eventual asymptotic convergence, but a direct and exponentially fast decay of disagreement from any initial condition, we must ensure that the stabilizing effect of consensus outweighs the destabilizing potential across the \emph{entire spectrum} of the system, not just for the mode corresponding to $\lambda_{N-1}$. This requires a stricter condition, as established in the following theorem.

\begin{theorem}[Exponential stability of the FSM]\label{thm:global_exponential_stability}
	Let us define $\alpha :=  K \max_{i \in \left\{1, \dots, N-1\right\}} \abs{\lambda_i}$. If $\alpha < 1$, then the fully synchronized manifold $\mathcal{S}$ is globally exponentially stable.
\end{theorem}

\begin{proof}
Let $\bm{e} = \bm{x} - \bar{x}\one$ be the error from center of mass of the agents $\bar{x} := \sum_{i=1}^N d_i x_i / \sum_{i=1}^N d_i$. When the error is zero, the agents are synchronized. The dynamics of the error is given by 	$\dot{\bm{e}} = \dot{\bm{x}} - \dot{\bar{x}} \one$.

Let us take the Lyapunov function $V(\bm{e}) = \norm{\bm{e}}_{\bm{D}}^2/2 = \bm{e}^\top \bm{D} \bm{e} / 2$ for which $V(\bm{e}) = 0$ if and only if $\bm{e} = 0$. The time derivative of the Lyapunov function is given by
\begin{align*}
	\dot{V}(\bm{e}) &= \bm{e}^\top \bm{D} \dot{\bm{e}} = \left( \bm{x} - \bar{x}\one\right)^\top \bm{D}\left( \bm{D}^{-1} \! \bm{A} \bm{s}(\bm{x}) - \bm{x} \right).
\end{align*}

Note that since $\one$ is a the right eigenvector of $\bm{L} = \bm{D} - \bm{A}$ associated to the eigenvalue $0$, one has
\begin{align*}
    \left( \bm{x} - \bar{x}\one\right)^\top \bm{A} \one &= \left( \bm{x} - \bar{x}\one\right)^\top \bm{D} \one = \bm{x}^\top \bm{D} \one - \bar{x} \one^\top \bm{D}\one\\
    &= \sum_{i=1}^N d_i x_i - \frac{\sum_{i=1}^N d_i }{ \sum_{i=1}^N d_i}\sum_{i=1}^N d_i x_i  = 0.
\end{align*}
Thus, $\bm{e}^\top \bm{A} \one = \bm{e}^\top \bm{D} \one = 0$. 

By substituting $\bm{x} = \bm{e} + \bar{x}\one$ and $\bm{s}(\bm{x}) = \bm{s}(\bm{x}) - s(\bar{x})\one + s(\bar{x})\one$ in the dynamics of the error, one has
\begin{align*}
	\dot{V}(\bm{e}) &= \bm{e}^\top \bm{A} \bm{s}(\bm{x}) - \bm{e}^\top \bm{D} \left(\bm{x} - \bar{x}\one + \bar{x}\one \right)\\
	&= \bm{e}^\top \bm{A} \left( \bm{s}(\bm{x}) - s(\bar{x})\one + s(\bar{x})\one\right) - \bm{e}^\top \bm{D} \bm{e}\\
	&= \bm{e}^\top \bm{A} \left( \bm{s}(\bm{x}) - s(\bar{x})\one \right) - \bm{e}^\top \bm{D} \bm{e}.
\end{align*}

The Lipschitz condition of $s$ implies $|s(x_i) - s(\bar{x})| \leq K |x_i - \bar{x}| = K |e_i|$, for all $i \in \left\{1, \dots, N\right\}$. This allows us to write:
\begin{equation*}
	\bm{s}(\bm{x}) - s(\bar{x})\one = K \bm{\Theta} \bm{e}.
\end{equation*}
where $\bm{\Theta} = \mathrm{diag}(\theta_1, \dots, \theta_N)$ is a diagonal matrix with state-dependent entries satisfying $\theta_i = (s(x_i) - s(\bar{x})) / K( x_i - \bar{x}) = (s(x_i) - s(\bar{x})) / K e_i$ for $e_i \neq 0$ and $\theta_i = 1$ for $e_i = 0$. Note that $0 \leq \theta_i \leq 1$.

Substituting this into the expression for $\dot{V}$ yields:
\begin{equation*}
	\dot{V}(\bm{e}) = K \bm{e}^\top \bm{A} \bm{\Theta} \bm{e} - \bm{e}^\top \bm{D} \bm{e}.
\end{equation*}
Let us perform a change of coordinates, $\tilde{\bm{e}} = \bm{D}^{1/2} \bm{e}$. The expression for $\dot{V}$ becomes:
\begin{align*}
	\dot{V}(\bm{e}) 
	&= K \tilde{\bm{e}}^\top \tilde{\bm{A}} \bm{\Theta} \tilde{\bm{e}} - \|\tilde{\bm{e}}\|_2^2.
\end{align*}
where $\tilde{\bm{A}} = \bm{D}^{-1/2} \bm{A} \bm{D}^{-1/2}$ is the normalized adjacency matrix. As in proof of Theorem~\ref{thm:only_sync_equilibria}, $\tilde{\bm{A}}$ is similar to $\bm{D}^{-1}\bm{A}$, therefore, they share the same eigenvalues $\lambda_1, \dots, \lambda_N$.

In the transformed coordinates, the normalized eigenvector of $\tilde{\bm{A}}$ for $\lambda_N=1$ is $\tilde{\bm{v}}_N = \bm{D}^{1/2}\one$. The transformed error vector $\tilde{\bm{e}}$ is orthogonal to $\tilde{\bm{v}}_N$:
\begin{equation*}
	\tilde{\bm{e}}^\top \tilde{\bm{v}}_N = (\bm{D}^{1/2}\bm{e})^\top \bm{D}^{1/2}\one = \bm{e}^\top \bm{D} \one = 0.
\end{equation*}
Thus, $\tilde{\bm{e}}$ lies in the subspace spanned by the eigenvectors $\{\tilde{\bm{v}}_1, \dots, \tilde{\bm{v}}_{N-1}\}$ of $\tilde{\bm{A}}$, which are orthogonal to $\tilde{\bm{v}}_N$.

We now bound the quadratic form by applying the Cauchy-Schwarz inequality:
\begin{equation*}
	\tilde{\bm{e}}^\top \tilde{\bm{A}} \bm{\Theta} \tilde{\bm{e}} = (\tilde{\bm{A}}\tilde{\bm{e}})^\top (\bm{\Theta}\tilde{\bm{e}}) \leq \|\tilde{\bm{A}}\tilde{\bm{e}}\|_2 \|\bm{\Theta}\tilde{\bm{e}}\|_2.
\end{equation*}
Let us bound the two norms separately.
Since $\bm{\Theta}$ is a diagonal matrix with entries $|\theta_i| \leq 1$, its induced 2-norm satisfies $\|\bm{\Theta}\|_2 = \max_i|\theta_i| \leq 1$. Thus, $\|\bm{\Theta}\tilde{\bm{e}}\|_2 \leq \|\tilde{\bm{e}}\|_2$.

For $\|\tilde{\bm{A}}\tilde{\bm{e}}\|_2$, we use the spectral decomposition of $\tilde{\bm{A}}$ since it is symmetric. Let $\{\tilde{\bm{v}}_i\}_{i=1}^N$ be an orthonormal basis of eigenvectors for $\tilde{\bm{A}}$. As established, $\tilde{\bm{e}}$ has no component along $\bm{v}_N$, so it can be written as $\tilde{\bm{e}} = \sum_{i=1}^{N-1} c_i \tilde{\bm{v}}_i$, where $c_i = \tilde{\bm{e}}^\top \tilde{\bm{v}}_i$. Then,
\begin{equation*}
	\|\tilde{\bm{A}}\tilde{\bm{e}}\|_2 = \Big\| \sum_{i=1}^{N-1} c_i \lambda_i \tilde{\bm{v}}_i \Big\|_2  \leq  \max_{j \in \{1, \dots, N-1\}} |\lambda_j| \|\tilde{\bm{e}}\|_2.
\end{equation*}

Combining these two bounds, we get:
\begin{equation*}
	\tilde{\bm{e}}^\top \tilde{\bm{A}} \bm{\Theta} \tilde{\bm{e}} \leq \Big( \max_{i \in \{1, \dots, N-1\}} |\lambda_i| \Big) \|\tilde{\bm{e}}\|_2^2.
\end{equation*}

Substituting this back into the expression for $\dot{V}(\bm{e})$:
\begin{align*}
	\dot{V}(\bm{e}) &\leq K \Big( \max_{i \in \{1, \dots, N-1\}} |\lambda_i| \Big) \|\tilde{\bm{e}}\|_2^2 - \|\tilde{\bm{e}}\|_2^2 \\
	&= 2 \Big( K \max_{i \in \{1, \dots, N-1\}} |\lambda_i| - 1 \Big) V(\bm{e}) = - 2 \alpha V(\bm{e}),
\end{align*}
since $\|\tilde{\bm{e}}\|_2^2 = 2V(\bm{e})$ and $\alpha = K \max_{i \in \{1, \dots, N-1\}} |\lambda_i| - 1 > 0$ by the condition of the theorem. By Grönwall's inequality, we have:
\begin{equation*}
	V(\bm{e}(t)) \leq V(\bm{e}(0)) e^{-2 \alpha t} \Leftrightarrow \|\bm{e}(t)\|_{\bm{D}} \leq \|\bm{e}(0)\|_{\bm{D}} e^{-\alpha t}.
\end{equation*}
This shows that $V(\bm{e}(t)) \to 0$ exponentially, which implies $\|\bm{e}(t)\|_{\bm{D}} \to 0$ exponentially. Since the graph is connected, $\bm{D}$ is positive definite, and the norm $\|\cdot\|_{\bm{D}}$ is equivalent to any other vector norm on $\mathbb{R}^N$. Thus, the manifold $\mathcal{S}$ (where $\bm{e}=0$) is globally exponentially stable.
\end{proof}

\section{Input-to-State Stability of Synchronized Clusters} \label{sec:ISS}

Having established the sharp conditions for full synchronization, we turn to the equally important and arguably more common phenomenon of \emph{local synchronization}, or clustering. In many complex networks, from biological systems to social organizations, global consensus is neither the goal nor the outcome. Instead, these systems self-organize into distinct and coherent clusters that perform specialized tasks while maintaining a persistent, structured disagreement with one another. This raises a natural question: why do such clusters maintain their internal cohesion when they are perpetually influenced by the rest of the network?

\subsection{Preliminaries on Input-to-State Stability}

This section analyzes the stability of these clusters through the lens of \emph{Input-to-State Stability}. We model the disagreement \emph{within} a cluster as the system state and the influence from the rest of the network as a persistent external \emph{input}. Our main result demonstrates that if a cluster possesses sufficient internal cohesion, its internal disagreement is bounded and robust to these external influences. Crucially, we show that this bound depends not on the total magnitude of the external influence, but on its heterogeneity across the cluster's nodes.

To build on this argument, we begin by recalling its formal definition.

\begin{definition}[Input-to-State Stable (ISS)]
	Consider a nonlinear system described by the differential equation
	\begin{equation*}
		\dot{\bm{x}}(t) = f(\bm{x}(t), \bm{u}(t)),
	\end{equation*}
	where $\bm{x}(t) \in \mathbb{R}^n$ is the state vector and $\bm{u}(t) \in \mathbb{R}^m$ is the input vector. Let $\bm{x}(t)$ denote the solution starting from the initial state $\bm{x}(0) = \bm{x}_0$ under the input $\bm{u}(t)$.
	
	The system is said to be \emph{Input-to-State Stable (ISS)} if there exist functions $\gamma: [0, \infty) \to [0, \infty)$ and $\beta: [0, \infty) \times [0, \infty) \to [0, \infty)$ such that:
	\begin{enumerate}
		\item $\gamma$ is a $\mathcal{K}$ function. i.e., it is continuous, strictly increasing, and $\gamma(0) = 0$.
		\item $\beta$ is a $\mathcal{K}\mathcal{L}$ function. i.e., for any fixed $t \ge 0$, $\beta(\cdot, t)$ is a $\mathcal{K}$ function, and for any fixed $r > 0$, $\beta(r, t)$ is decreasing with respect to $t$ and $\lim_{t \to \infty} \beta(r, t) = 0$.
	\end{enumerate}

	And for any initial state $\bm{x}_0 \in \mathbb{R}^n$ and any bounded measurable input $\bm{u}: [0, \infty) \to \mathbb{R}^m$, the solution $\bm{x}(t)$ exists for all $t \ge 0$ and satisfies the inequality:
	\begin{equation*} \label{eq:iss_definition}
		\norm{\bm{x}(t)} \le \beta(\norm{\bm{x}_0}, t) + \gamma\big(\sup_{0 \leq \tau \leq t} \norm{\bm{u}(\tau)}\big), \quad \text{for all $t \ge 0$.}
	\end{equation*}
\end{definition}
	
\begin{remark}	
	The ISS inequality captures two key properties. First, with zero input ($u(t)\equiv0$), the system is globally asymptotically stable. Second, with a bounded input, the state is guaranteed to remain ultimately bounded. For finite-dimensional systems, this property holds irrespective of the chosen norm.
\end{remark}

\subsection{ISS Analysis of Internal Cluster Disagreement}

We now apply this framework to analyze the dynamics within a connected subgraph $\Gcal' = (\Vcal', \Ecal')$ that contains $N'$ nodes. Our goal is to express the dynamics of the cluster $\Vcal'$ in the form of an ISS system, where the ``state" is the internal disagreement within the cluster and the ``input" is the influence from the rest of the network. i.e., from the node in $\Vcal \setminus \Vcal'$.

Let $\bm{S}\in \R^{N'\times N}$ be the selection matrix that maps the full state vector $\bm{x}$ to the subgraph state vector $\bm{x}' := \bm{S} \bm{x}$. Specifically, for $j \in \Vcal'$, $\bm{S}_{ij} = 1$ for its corresponding index $i \in \Vcal$ and $0$ otherwise. Note that $\bm{S}\bm{S}^\top = \bm{I}_{N'}$ and $\bm{S}^\top\bm{S} = \mathrm{diag}(\mathbb{I}_{\Vcal'})$, where $\mathbb{I}_{\Vcal'}\in \R^N$ is the indicator vector of $\Vcal' \subseteq \Vcal$. i.e., $[\mathbb{I}_{\Vcal'}]_i = 1$ if node $i \in \Vcal$ is in $\Vcal'$ and $0$ otherwise.

Applying this to the system dynamics \eqref{eq:dynamic} yields:
\begin{align} \label{eq:subgraph_dyn_intermediate}
	\dot{\bm{x}}' = \bm{S} \bm{D}^{-1} \bm{A} \bm{s}(\bm{x}) - \bm{x}'.
\end{align}

We can decompose the term $\bm{A} \bm{s}(\bm{x})$ based on the subgraph structure. Let $\bm{A}_{\mathrm{in}} = \bm{S} \bm{A} \bm{S}^\top \in \R^{N\times N} $ be the adjacency matrix of the induced subgraph $\Gcal'$. Let $\bm{A}_{\mathrm{ext}}$ be the matrix representing connections between $\Vcal'$ and $\Vcal \setminus \Vcal'$. Then, for the nodes in $\Vcal'$, the input from the network can be split into an internal and an external part:
\begin{equation*}
    \bm{S} \bm{A} \bm{s}(\bm{x}) = \bm{A}_{\mathrm{in}} \bm{s}(\bm{x}') + \bm{S}\bm{A}_{\mathrm{ext}} \bm{s}(\bm{x}).
\end{equation*}

Let us define the following:
\begin{itemize}
	\item $\bm{D}_{\mathrm{in}} = \mathrm{diag}(d_i^{\mathrm{in}})_{i \in \Vcal'}$, where $d_i^{\mathrm{in}} = \sum_{j \in \Vcal'} a_{ij}$, is the diagonal matrix of internal degrees for nodes in $\Gcal'$.
	\item $\bm{\widetilde{D}}_{\mathrm{in}} = \bm{S} \bm{D} \bm{S}^\top \in \R^{N'\times N'}$ is the diagonal matrix containing the total degrees $d_i$ for nodes $i \in \Vcal'$ in the original graph $\Gcal$. Note that $\bm{\widetilde{D}}_{\mathrm{in}}^{-1} = \bm{S} \bm{D}^{-1} \bm{S}^\top$. Since $\bm{S}^\top\bm{S}$ is diagonal and $\bm{S}\bm{S}^\top = \bm{I}_{N'}$, one has $\bm{\widetilde{D}}_{\mathrm{in}}^{-1} \bm{S} = \bm{S} \bm{D}^{-1} \bm{S}^\top\bm{S}=  \bm{S} \bm{S}^\top\bm{S} \bm{D}^{-1}  = \bm{S} \bm{D}^{-1}$.
\end{itemize}

Substituting these definitions into \eqref{eq:subgraph_dyn_intermediate} and rearranging allows us to express the dynamics as an ideal internal system plus two distinct perturbation terms:
\begin{align} \label{eq:dynamics_subgraph_perturbated}
	\dot{\bm{x}}' &= \bm{D}_{\mathrm{in}}^{-1} \bm{A}_{\mathrm{in}} \bm{s}(\bm{x}') - \bm{x}' \nonumber \\
	& \qquad + \left( \bm{\widetilde{D}}_{\mathrm{in}}^{-1} - \bm{D}_{\mathrm{in}}^{-1} \right) \bm{A}_{\mathrm{in}} \bm{s}(\bm{x}') + \bm{\widetilde{D}}_{\mathrm{in}}^{-1} \bm{S} \bm{A}_{\mathrm{ext}} \bm{s}(\bm{x}) \nonumber \\
	& := \underbrace{\left(\bm{D}_{\mathrm{in}}^{-1} \bm{A}_{\mathrm{in}} \bm{s}(\bm{x}') - \bm{x}'\right)}_{\text{Ideal Internal Dynamics}} + \underbrace{\bm{p}(\bm{x}).}_{\text{Total Perturbation}}
\end{align}

The total perturbation $\bm{p}(\bm{x})$ originates from two sources:
\begin{itemize}
    \item \emph{External Perturbation}: Influence from neighbors outside the cluster, given by $\bm{\widetilde{D}}_{\mathrm{in}}^{-1} \bm{S} \bm{A}_{\mathrm{ext}} \bm{s}(\bm{x})$.
    \item \emph{Internal Perturbation}: A more subtle effect arising from the mismatch between a node's total degree and its internal degree, given by $\left( \bm{\widetilde{D}}_{\mathrm{in}}^{-1} - \bm{D}_{\mathrm{in}}^{-1} \right) \bm{A}_{\mathrm{in}} \bm{s}(\bm{x}')$.
\end{itemize}

A key insight of consensus dynamics is that the system's internal disagreement is \emph{insensitive to any uniform perturbation that pushes all nodes in the cluster equally}. The component of the perturbation that drives disagreement is its non-uniform part. We therefore define the \emph{residual perturbation} as the component of $\bm{p}(\bm{x})$ that is orthogonal to the fully synchronized manifold of the cluster:
\begin{equation*}
	\bm{\tilde{p}}(\bm{x}) := \bm{p}(\bm{x}) - \bar{p}(\bm{x}) \one_{N'},
\end{equation*}
where $\bar{p}(\bm{x}) = (\sum_{i\in \Vcal'} d_i^{\mathrm{in}} [p(\bm{x})]_i)/(\sum_{i \in \Vcal'} d_i^{\mathrm{in}})$ is the weighted average of the perturbation over the cluster.

The norm of the residual perturbation, 
\begin{equation*}
    \norm{\bm{\tilde{p}}(\bm{x})}_{\bm{D}_{\mathrm{in}}}^2 = \sum_{i=1}^{N'} d_i^{\mathrm{in}}([p(\bm{x})]_i - \bar{p}(\bm{x}))^2,
\end{equation*}
is a weighted variance of the total perturbation's effect on the subgraph. The closest $\norm{\bm{\tilde{p}}(\bm{x})}_{\bm{D}_{\mathrm{in}}}$ is to $0$, the most uniform is the perturbation across the subgraph $\Gcal'$. We can define an upper bound for this norm as:
\begin{equation*}
	\norm{\bm{\tilde{p}}(\bm{x})}_{\bm{D}_{\mathrm{in}}} \leq \sup_{\bm{x} \in \mathcal{X}} \norm{\bm{\tilde{p}}(\bm{x})}_{\bm{D}_{\mathrm{in}}} := P^{\mathrm{sup}}.
\end{equation*}

We can now state the main result of this section.

\begin{theorem}[ISS of Synchronized Clusters]\label{thm:ISS_subgraph}
Consider a connected induced subgraph $\Gcal' = (\Vcal', \Ecal')$ of the graph $\Gcal$. Let us define $\alpha_{\mathrm{in}} := 1 - K \max_{i\in \{1, \dots, N-1\}}| \lambda_{i}(\bm{D}_{\mathrm{in}}^{-1} \bm{A}_{\mathrm{in}})|$.

If the internal stability condition $\alpha_{\mathrm{in}} > 0$ holds, then the disagreement within the subgraph $\Gcal'$ is input-to-state stable with respect to the residual perturbation $\bm{\tilde{p}}(\bm{x})$.

Specifically, for the disagreement vector $\bm{e} = \bm{x}' - \bar{x}' \one_{N'}$, where $\bar{x}' = (\sum_{i \in \Vcal'} d_i^{\mathrm{in}} x_i)/(\sum_{i \in \Vcal'} d_i^{\mathrm{in}})$, there exist a $\mathcal{KL}$ function $\beta$ and a $\mathcal{K}$ function $\gamma$, given by
	\begin{equation*}
		\beta(r, t) = r e^{-\alpha_{\mathrm{in}} t} \quad \text{and} \quad \gamma(r) = \frac{r}{\alpha_{\mathrm{in}}},
	\end{equation*}
	such that the following ISS inequality holds for all $t \ge 0$:
	\begin{equation*}
		\norm{\bm{e}(t)}_{\bm{D}_{\mathrm{in}}} \leq \beta(\norm{\bm{e}(0)}_{\bm{D}_{\mathrm{in}}}, t) + \gamma\left(\sup_{0 \leq \tau \leq t} \norm{\bm{\tilde{p}}(\bm{x}(\tau))}_{\bm{D}_{\mathrm{in}}}\right).
	\end{equation*}	

Moreover, the norm of the disagreement $\norm{\bm{e}(t)}_{\bm{D}_{\mathrm{in}}}$ is ultimately bounded by $P^{\mathrm{sup}} / \alpha_{\mathrm{in}}$.	
\end{theorem}

\begin{proof}
	The proof starts by considering the ISS-Lyapunov function candidate:
	\begin{equation*}
		V(\bm{e}) = \frac{1}{2} \bm{e}^\top \bm{D}_{\mathrm{in}} \bm{e} = \frac{1}{2} \norm{\bm{e}}_{\bm{D}_{\mathrm{in}}}^2.
	\end{equation*}

	From the dynamics in \eqref{eq:dynamics_subgraph_perturbated} and following similar steps to the proof of Theorem~\ref{thm:global_exponential_stability}, the time derivative of $V(\bm{e})$ along the trajectories of the system is:
	\begin{equation*}
		\dot{V}(\bm{e}) = \bm{e}^\top \bm{A}_{\mathrm{in}} \bm{s}(\bm{x}') - \bm{e}^\top \bm{D}_{\mathrm{in}} \bm{x}' + \bm{e}^\top \bm{D}_{\mathrm{in}} \bm{p}(\bm{x}).
	\end{equation*}
	The internal dynamics term is bounded as (from the proof of Theorem~\ref{thm:global_exponential_stability}):
	\begin{equation*}
		\bm{e}^\top \bm{A}_{\mathrm{in}} \bm{s}(\bm{x}') - \bm{e}^\top \bm{D}_{\mathrm{in}} \bm{x}' \leq  -2\alpha_{\mathrm{in}} V(\bm{e}).
	\end{equation*}

	The perturbation term $\bm{e}^\top \bm{D}_{\mathrm{in}} \bm{p}(\bm{x})$ is analyzed by decomposing the perturbation vector $\bm{p}(\bm{x})$ into its $\bm{D}_{\mathrm{in}}$-weighted average, $\bar{p}(\bm{x}) \one_{N'}$, and the residual part, $\bm{\tilde{p}}(\bm{x})$. Thus, $\bm{p}(\bm{x}) = \bar{p}(\bm{x}) \one_{N'} + \bm{\tilde{p}}(\bm{x})$.
		
	Substituting this decomposition into the perturbation term yields:
	\begin{align*}
		\bm{e}^\top \bm{D}_{\mathrm{in}} \bm{p} (\bm{x}) &= \bm{e}^\top \bm{D}_{\mathrm{in}} \left( \bar{p}(\bm{x}) \one_{N'} + \bm{\tilde{p}}(\bm{x}) \right) \\
        &= \bar{p}(\bm{x})\bm{e}^\top \bm{D}_{\mathrm{in}}  \one_{N'} + \bm{e}^\top \bm{D}_{\mathrm{in}}\bm{\tilde{p}}(\bm{x})= \bm{e}^\top \bm{D}_{\mathrm{in}} \bm{\tilde{p}}(\bm{x}).
	\end{align*}
	Since $\bm{e}^\top \bm{D}_{\mathrm{in}} \one_{N'} = 0$ by definition of the weighted average $\bar{x}'$, the first term vanishes (as established in proof of Theorem~\ref{thm:global_exponential_stability}).
	
	Using the Cauchy-Schwarz inequality for the $\bm{D}_{\mathrm{in}}$-weighted inner product $\langle \mathbf{u}, \mathbf{v} \rangle_{\bm{D}_{\mathrm{in}}} = \mathbf{u}^\top \bm{D}_{\mathrm{in}} \mathbf{v}$:
	\begin{align*}
		\bm{e}^\top \bm{D}_{\mathrm{in}} \bm{\tilde{p}}(\bm{x}) = \langle \bm{e}, \bm{\tilde{p}}(\bm{x}) \rangle_{\bm{D}_{\mathrm{in}}} \leq  \norm{ \bm{e} }_{\bm{D}_{\mathrm{in}}} \norm{ \bm{\tilde{p}}(\bm{x}) }_{\bm{D}_{\mathrm{in}}}.
	\end{align*}

Combining these, the differential inequality for $V(\bm{e})$ becomes:
\begin{equation*}
	\dot{V}(\bm{e}) \leq -2\alpha_{\mathrm{in}} V(\bm{e}) + \norm{\bm{e}}_{\bm{D}_{\mathrm{in}}} \norm{\bm{\tilde{p}}(\bm{x}(t))}_{\bm{D}_{\mathrm{in}}}.
\end{equation*}

Let $y(t) = \norm{\bm{e}(t)}_{\bm{D}_{\mathrm{in}}}$. Then $V(\bm{e}(t)) = \frac{1}{2} y(t)^2$ and its time derivative is $\dot{V}(t) = y(t) \dot{y}(t)$. Substituting these into the inequality gives:
\begin{align*}
	y(t) \dot{y}(t) \leq -\alpha_{\mathrm{in}} y(t)^2 + y(t) \norm{\bm{\tilde{p}}(\bm{x}(t))}_{\bm{D}_{\mathrm{in}}}.
\end{align*}

For any $t$ where $y(t) > 0$, i.e., $\norm{\bm{e}(t)}_{\bm{D}_{\mathrm{in}}} > 0$, we can divide by $y(t)$:
\begin{equation*}
	\dot{y}(t) \leq -\alpha_{\mathrm{in}} y(t) + \norm{\bm{\tilde{p}}(\bm{x}(t))}_{\bm{D}_{\mathrm{in}}} = -\alpha_{\mathrm{in}} y(t) + u(t),
\end{equation*}
where $u(t) = \norm{\bm{\tilde{p}}(\bm{x}(t))}_{\bm{D}_{\mathrm{in}}}$. This is a linear differential inequality $\dot{y}(t) \leq -\alpha_{\mathrm{in}} y(t) + u(t)$, which also holds trivially when $y(t)=0$.

Applying Gronwall's inequality, the solution is bounded by:
\begin{align*}
	y(t) &\leq y(0) e^{-\alpha_{\mathrm{in}} t} + \int_0^t e^{-\alpha_{\mathrm{in}}(t-\tau)} u(\tau) \mathrm{d}\tau \\
	&\leq y(0) e^{-\alpha_{\mathrm{in}} t} + \left(\sup_{0 \leq \tau \leq t} u(\tau)\right) \int_0^t e^{-\alpha_{\mathrm{in}}(t-\tau)} \mathrm{d}\tau \\
	&= y(0) e^{-\alpha_{\mathrm{in}} t} + \frac{1}{\alpha_{\mathrm{in}}} \left(\sup_{0 \leq \tau \leq t} u(\tau)\right) (1 - e^{-\alpha_{\mathrm{in}} t}).
\end{align*}
Since $(1 - e^{-\alpha_{\mathrm{in}} t}) \le 1$ for $t \ge 0$, we have the simpler bound:
\begin{equation*}
	y(t) \leq y(0) e^{-\alpha_{\mathrm{in}} t} + \frac{1}{\alpha_{\mathrm{in}}} \sup_{0 \leq \tau \leq t} u(\tau).
\end{equation*}

Substituting back $y(t) = \norm{\bm{e}(t)}_{\bm{D}_{\mathrm{in}}}$ yields the ISS inequality:
\begin{equation*}
	\norm{\bm{e}(t)}_{\bm{D}_{\mathrm{in}}} \leq \norm{\bm{e}(0)}_{\bm{D}_{\mathrm{in}}} e^{-\alpha_{\mathrm{in}} t} + \frac{1}{\alpha_{\mathrm{in}}} \sup_{0 \leq \tau \leq t} \norm{\bm{\tilde{p}}(\bm{x}(\tau))}_{\bm{D}_{\mathrm{in}}}.
\end{equation*}
with $\beta(r, t) = r e^{-\alpha_{\mathrm{in}} t}$ and $\gamma(r) = r /\alpha_{\mathrm{in}}$.

For the ultimate bound, taking the limit sup as $t \to \infty$:
\begin{align*}
	\limsup_{t \to \infty} \norm{\bm{e}(t)}_{\bm{D}_{\mathrm{in}}} \leq \limsup_{t \to \infty} \Big( \norm{\bm{e}(0)}_{\bm{D}_{\mathrm{in}}} e^{-\alpha_{\mathrm{in}} t} + \\ \frac{1}{\alpha_{\mathrm{in}}} \sup_{0 \leq \tau \leq t} \norm{\bm{\tilde{p}}(\bm{x}(\tau))}_{\bm{D}_{\mathrm{in}}} \Big).
\end{align*}
Since $\alpha_{\mathrm{in}} > 0$, the first term goes to zero, leaving:
\begin{align*}
	\limsup_{t \to \infty} \norm{\bm{e}(t)}_{\bm{D}_{\mathrm{in}}} \leq \frac{1}{\alpha_{\mathrm{in}}} \sup_{t \ge 0} \norm{\bm{\tilde{p}}(\bm{x}(t))}_{\bm{D}_{\mathrm{in}}} \leq \frac{P^{\mathrm{sup}}}{\alpha_{\mathrm{in}}}.
\end{align*}
This concludes the proof.
\end{proof}

Theorem~\ref{thm:ISS_subgraph} states that a sufficiently cohesive cluster can maintain its internal agreement, even when perpetually influenced by the rest of the network. The theorem formalizes this by demonstrating that the cluster's internal disagreement is ISS. This implies that while initial disagreement within the cluster decays exponentially, persistent external influence results in a finite ultimate bound on this disagreement. Thus, Theorem~\ref{thm:ISS_subgraph} provides a tool for analyzing the stability of NFSE, interpreting the network as an interconnection of robust and internally stable clusters.

An important consequence of this framework is the identification of \emph{variance} (or non-uniformity) of the external perturbation being the true cause of internal disagreement. A uniform influence, regardless of its magnitude, would not prevent the cluster from achieving perfect internal consensus asymptotically. It is therefore the \emph{heterogeneity} of external signals, which causes boundary nodes to be pulled in different directions, that fundamentally limits local synchronization. This explains the clustering obtained in \cite{couthuresAnalysisOpinionDynamics2024}, where a uniform environmental signal becomes an effectively non-uniform perturbation due to heterogeneous agent gains, sustaining disagreement as our framework predicts. 

\begin{remark}[A Unifying Principle of Stability]
	The condition for a cluster's internal stability, $\alpha_{\mathrm{in}} > 0$, is structurally identical to the condition for global exponential stability in Theorem~\ref{thm:global_exponential_stability}. This reveals a scale-free principle: a group of agents, whether the entire network or a single cluster, maintains cohesion if its internal connectivity is strong enough to overcome the maximal destabilizing effect of the nonlinearity, a principle elegantly quantified by the eigenvalues of the relevant subgraph dynamics.
\end{remark}

\begin{remark}[A Practical Bound on Perturbation]
	The ultimate bound on disagreement depends on the maximum residual perturbation, $P^{\mathrm{sup}}$. This term can be bounded by the structural properties of the cluster's boundary:
    \begin{equation*}
        P^{\mathrm{sup}} \leq 2 \sum_{i \in \Vcal'} d_i^{\mathrm{in}} \left(\frac{d_i^{\mathrm{ext}}}{d_i}\right)^2,
    \end{equation*}
    where $d_i^{\mathrm{ext}} = d_i - d_i^{\mathrm{in}}$ is the number of external connections for node $i \in \Vcal'$. This provides a concrete, computable upper bound on the asymptotic synchronization error within a cluster, based entirely on the network topology.
\end{remark}

\begin{figure}
	\centering
	\begin{subfigure}[t]{\linewidth}
		\centering
		\includegraphics[width=.75\linewidth]{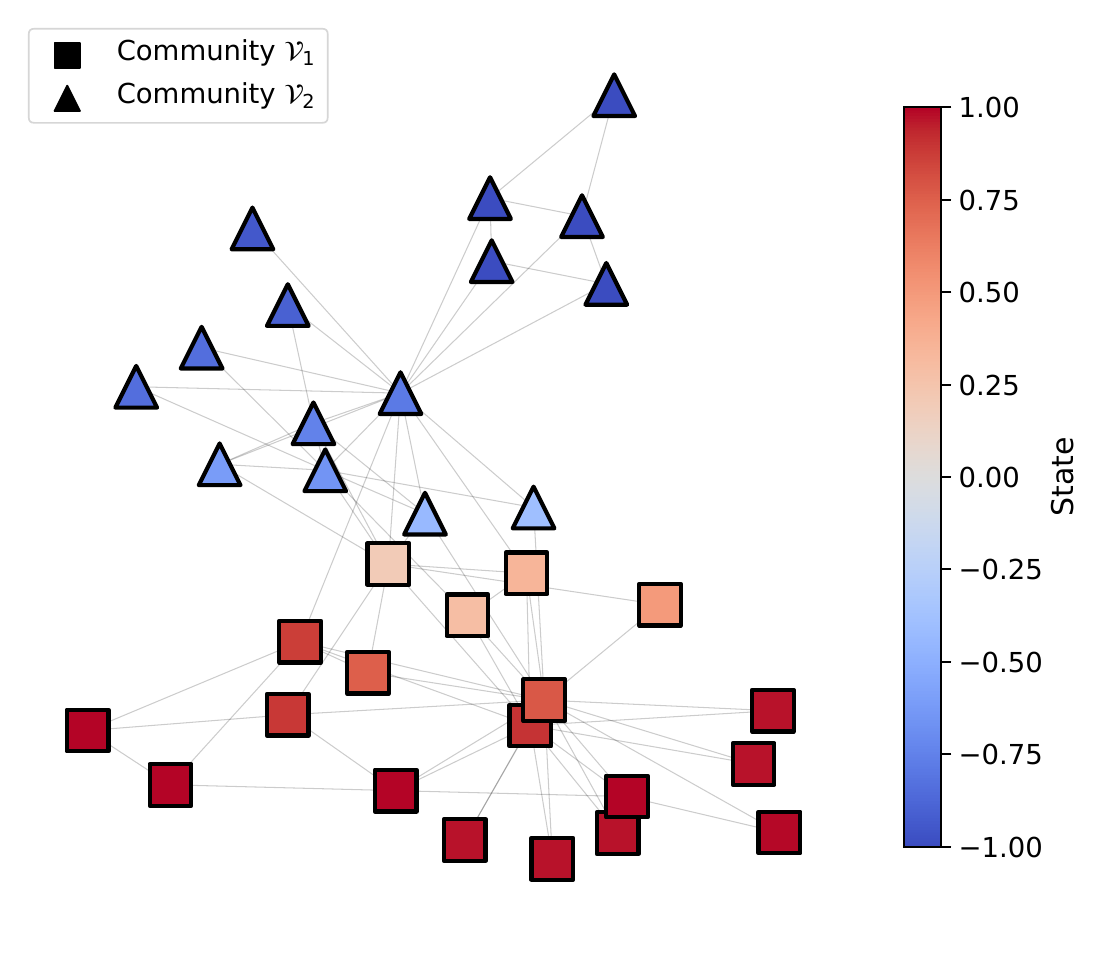}
		\caption{Final equilibrium state of the network. Node shapes (triangle/square) indicate the two strong communities: $\mathcal{V}_1$ and $\mathcal{V}_2$. The communities were obtained using \cite{MG10}. Node colors represent the final state value $x_i^*$, demonstrating local synchronization within clusters but global disagreement between them.}
        \label{fig:karate_a}
	\end{subfigure}
    \hfil
	\begin{subfigure}[t]{\linewidth}
		\centering
		\includegraphics[width=\linewidth]{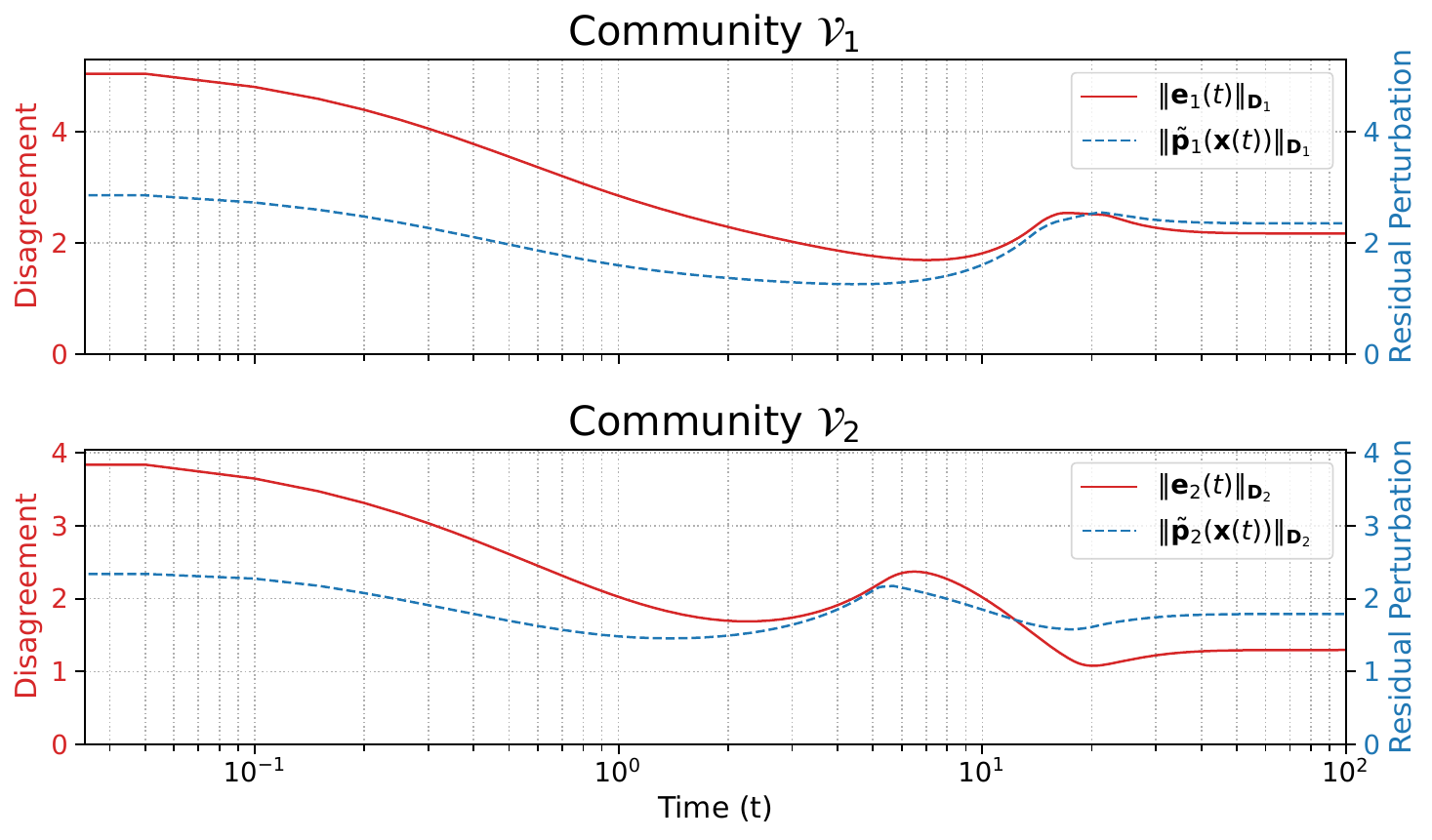}
		\caption{ISS analysis for each community. For $k \in \{1,2\}$, the internal disagreement $\norm{\bm{e}_k(t)}_{\bm{D}_k}$ (the ISS ``state" in red) decays exponentially from its initial condition and converges to a non-zero ultimate bound. This bound is sustained by the persistent residual perturbation $\norm{\bm{\tilde{p}}_k(\bm{x}(t))}_{\bm{D}_k}$ from the opposing cluster (the ISS ``input" in blue), precisely as predicted by Theorem~\ref{thm:ISS_subgraph}.}
        \label{fig:karate_b}
	\end{subfigure}
    \caption{Numerical validation of the ISS framework on Zachary's Karate Club network, a canonical example of a social network with a strong, empirically verified community structure. The simulation uses a signal function $s(x)= \max(-1, \min(1, Kx))$ (with $K=1.2$) chosen to violate the threshold for full synchronization ($K\lambda_{N-1} > 1$), forcing the emergence of a stable, clustered equilibrium (NFSE).}\vspace{-.5cm}
    \label{fig:karate}
\end{figure}

\subsection{Numerical Illustration: Robust Clustering in the Karate Club Network}

To demonstrate our theoretical framework in action, we apply it to Zachary's Karate Club network \cite{zacharyInformationFlowModel1977}. As detailed in the caption of Figure~\ref{fig:karate}, the simulation parameters are chosen specifically to violate the conditions for full synchronization, creating an environment where clustering is the expected outcome.

As predicted by our analysis, the network forgoes global consensus and instead converges to a stable NFSE. The final state of the network, visualized in Figure~\ref{fig:karate_a}, shows a clear partitioning of agents' opinions. The states align perfectly with the two ground-truth communities, forming two internally synchronized but mutually opposed clusters. This outcome confirms that, when the sharp threshold is crossed, the system can transition into a state of structured, local agreement.

This clustered equilibrium provides a perfect setting for validating the core predictions of our ISS analysis. In Figure~\ref{fig:karate_b}, we treat each community as a subsystem and plot the time evolution of its internal disagreement (the ISS ``state") and the residual perturbation it receives from the opposing cluster (the ISS ``input"). The resulting dynamics are a standard illustration of Input-to-State Stability:
\begin{itemize}
    \item The internal disagreement (red curve) exponentially decays from its initial value, demonstrating the cluster's internal stability ($\alpha_{\mathrm{in}} > 0$).
    \item However, the disagreement does not converge to zero. It is sustained by the persistent, non-uniform perturbation from the other cluster (blue curve), settling to a non-zero ultimate bound.
\end{itemize}
This behavior, of exponential decay to a non-zero bound sustained by a persistent input, is precisely the outcome guaranteed by the ISS inequality in Theorem~\ref{thm:ISS_subgraph}. The simulation thus provides a concrete validation of our theoretical results.

\section{Conclusion}\label{sec:conclusion}

This work established a comprehensive theory for the emergence and stability of synchronized states in multi-agent systems with general nonlinear interactions. We demonstrated that the transition from global consensus to robust, localized clustering is governed by a sharp threshold, $K\lambda_{N-1} < 1$, that unifies the interplay between agent nonlinearity and network connectivity. By introducing a novel framework based on Input-to-State Stability, we moved beyond analyzing the existence of these clustered states to formally quantifying their robustness. This analysis yielded a practical insight: the cohesion of a synchronized cluster is fundamentally limited not by the \emph{magnitude}, but by the \emph{heterogeneity} of external network perturbations.

These findings provide both a predictive design principle and a formal methodology for analyzing modularity in complex networks, opening several compelling avenues for future research. An immediate next step is to deepen the analysis of the clusters themselves, moving beyond internal cohesion to performance guarantees such as robust reference tracking and state containment.




\appendix

\subsection{Proof of Lemma~\ref{lemma:existence_of_stable_fixed_point}}\label{app:proof_lemma2}

\begin{proof}
    The fixed points of $s$ are the zeros of the continuous function $g(x) := s(x) - x$ such that $g(-1) \geq 0$ and $g(1) \leq 0$. By the Intermediate Value Theorem, $\mathrm{Fix}(s)$, is a non-empty closed subset of $[-1, 1]$.	If $\mathrm{Fix}(s)$ contains an interval, any point in its interior is trivially a stable fixed point. We therefore assume that $\mathrm{Fix}(s)$ consists only of isolated points. As a compact set, $\mathrm{Fix}(s)$ must be finite. Let the fixed points be ordered as $\mathrm{Fix}(s) = \{c_1, c_2, \dots, c_M\}$, where $-1 \leq c_1 < \dots < c_M \leq 1$. The following handle cases of signal functions similar to the one in Figure~\ref{fig:fs_dynamics_examples:b}.

	We proceed by contradiction, assuming that all fixed points $c_i$ are unstable. i.e., there exists an $x$ such that $(x - c_i) g(x) > 0$. Consider the smallest fixed point, $c_1$. On $[-1, c_1)$, if non-empty, $g(x)$ must be positive since $g(-1) \ge 0$. For any $x$ in this interval, $(x - c_1) g(x) < 0$, so the stability condition is satisfied to the left of $c_1$. For $c_1$ to be unstable, the violation must occur on its right, which requires $g(x) > 0$ for $x$ in some interval $(c_1, c_1+\delta)$. By continuity, this implies $g(x) > 0$ for all $x \in (c_1, c_2)$. By induction, $g(x) > 0$ on every interval $(c_i, c_{i+1})$ for $i=1, \dots, M-1$.

	Finally, we examine the largest fixed point, $c_M$. Using the induction result, $g(x) > 0$ on $(c_{M-1}, c_M)$, which implies that the stability condition $(x - c_M) g(x) \le 0$ holds on the left of $c_M$. Since $c_M$ is assumed to be unstable, there must be a violation to its right, in the interval $(c_M, 1]$. This requires $g(x) > 0$ for some $x \in (c_M, 1]$. If $c_M=1$, the interval $(c_M, 1]$ is empty, so no violation can occur, contradicting our assumption. If $c_M < 1$, by continuity, $g(1) = \lim_{x \to 1^-} g(x) \ge 0$. But we know $g(1) \le 0$. Thus, $g(1)=0$, which means that $1$ is a fixed point, contradicting $c_M$ being the largest fixed point.

	In all cases, the assumption that all fixed points are unstable is absurd. Therefore, at least one stable fixed point exists.
\end{proof}

\subsection{Proof of Proposition~\ref{prop:non_sync_equilibrium}-\ref{prop:non_sync_equilibrium_2}}\label{app:proof_prop5}

\begin{proof}
    \ref{prop:non_sync_equilibrium_2} Assume $K \lambda_{N-1} > 1$. Under Assumption~\ref{ass:graph}, one has $\lambda_{N-1} < 1$, which implies $K > 1/\lambda_{N-1} > 1$. We proceed by contradiction. Assume that for \emph{any} signal function $s$ satisfying Assumption~\ref{ass:signal} with Lipschitz constant $K$, all equilibria of \eqref{eq:dynamic} are FSE.

    We construct a specific signal function $s$ that satisfies these conditions yet leads to a contradiction. Let $s:[-1,1] \to [-1,1]$ be a smooth, strictly increasing, and odd function such that $s'(0) = K$. An example of such a function can be constructed based on $s(x) = \tanh(K x)$ illustrated in Figure~\ref{fig:fs_dynamics_examples:a}. Since $s$ is odd and strictly increasing with $s'(0) = K > 1$, its fixed points are $\mathrm{Fix}(s) = \{-c, 0, c\}$ for some $c \in (0, 1]$. By Theorem~\ref{thm:local_stability_of_synchronization_equilibria}, the FSE $\bm{x}^*_{-c} = -c\one$ and $\bm{x}^*_{c} = c\one$ are locally asymptotically stable, while the FSE $\bm{0}$ is unstable.

    Since $s$ is strictly increasing and smooth, by Lemma~\ref{lemma:SOP_flow}, every trajectory must converge to one of the three FSE: $-c\one$, $\bm{0}$, or $c\one$. This implies a partition of the state space $\Xcal$ into the basins of attraction of these three equilibria:
    \begin{equation*}
        \Xcal = \mathcal{B}(-c\one) \cup \mathcal{B}(\bm{0}) \cup \mathcal{B}(c\one).
    \end{equation*}
    The basins $\mathcal{B}(-c\one)$ and $\mathcal{B}(c\one)$ are non-empty, disjoint open sets. As the vector field $\bm{f}(\bm{x}) = \bm{D}^{-1}\bm{A}\bm{s}(\bm{x}) - \bm{x}$ is odd due to the oddness of $s$, the flow is symmetric with respect to the origin, implying $\mathcal{B}(-c\one) = -\mathcal{B}(c\one)$, as observed on the restriction to $\Scal$ in Figure~\ref{fig:fs_dynamics_examples:a}. Consequently, their common boundary, the \emph{separatrix}, is a single set $\partial \mathcal{B}(c\one) = \partial \mathcal{B}(-c\one)$.
    
    We now apply \cite[Theorem~4.8]{chiangStabilityRegionsNonlinear2015d} to characterize this separatrix. We must verify its assumptions:
    \begin{itemize}
        \item[(A1)] Hyperbolicity: The equilibrium points on the boundary $\partial \mathcal{B}(c\one)$ must be hyperbolic.
        \item[(A2)] Transversality: The stable and unstable manifolds of equilibrium points on the stability boundary satisfy the transversality condition.
        \item[(A3)] Boundary Convergence: Every trajectory on $\partial \mathcal{B}(c\one)$ must approach one of the equilibrium points on the boundary.
    \end{itemize}
    Assumption (A2) can be satisfied by a smooth choice of $s$, as $s(x) = \tanh(K x)$. To verify (A1) and (A3), we first identify the equilibria on the boundary. Let $\bm{y} \in \partial \mathcal{B}(c\one)$. Since the boundary of an invariant set is itself invariant, the trajectory $\bm{y}(t)$ remains in $\partial \mathcal{B}(c\one)$. As all trajectories converge to an equilibrium, $\bm{y}(t)$ must converge to one of $\{-c\one, \bm{0}, c\one\}$. It cannot converge to $c\one$ (as $\bm{y}$ would be in the interior of the basin) or $-c\one$ (as the basins are disjoint). Thus, any trajectory on the boundary must converge to $\bm{0}$. This implies that $\bm{0}$ is the only equilibrium point on the boundary $\partial \mathcal{B}(c\one)$, satisfying (A3). The Jacobian at the origin is $\bm{J}(\bm{0}) = s'(0) \bm{D}^{-1}\bm{A} - \bm{I} = K \bm{D}^{-1}\bm{A} - \bm{I}$. Its eigenvalues are $K\lambda_i - 1$, which are non-zero for a generic choice of $K$, satisfying (A1).
    
    With the assumptions verified, \cite[Theorem~4.8~(b)]{chiangStabilityRegionsNonlinear2015d} yields:
        $\partial \mathcal{B}(c\one) = W^s(\bm{0})$,
    where $W^s(\bm{0})$ is the stable manifold of $\bm{0}$.
    
    This leads to a topological contradiction based on dimensionality. The set $\partial \mathcal{B}(c\one)$ is the boundary separating two open sets in the $N$-dimensional space $\Xcal$, and thus it must be an $(N-1)$-dimensional manifold. However, the dimension of the stable manifold $W^s(\bm{0})$ is equal to the number of eigenvalues of $\bm{J}(\bm{0})$ with negative real parts. The eigenvalues are $\nu_i = K\lambda_i - 1$ for $i=1,\dots,N$.
    By assumption, $K \lambda_{N-1} > 1$. We also know $\lambda_N = 1$, and since $K > 1/\lambda_{N-1} > 1$, we have:
    \begin{itemize}
        \item $\nu_N = K\lambda_N - 1 = K - 1 > 0$.
        \item $\nu_{N-1} = K\lambda_{N-1} - 1 > 0$.
    \end{itemize}
    Thus, at least two eigenvalues of $\bm{J}(\bm{0})$ have positive real parts. The dimension of the unstable manifold, $\dim(W^u(\bm{0}))$, is therefore at least 2. Consequently, the dimension of the stable manifold is    $\dim(W^s(\bm{0})) = N - \dim(W^u(\bm{0})) \leq N-2$.
    
    We then have a contradiction: a manifold of dimension at most $N-2$ cannot separate an $N$-dimensional space. i.e.,
        $N-1 = \dim(\partial \mathcal{B}(c\one)) = \dim(W^s(\bm{0})) \leq N-2$.
    This contradiction implies that our initial assumption was false. Therefore, for our constructed function $s$, there must exist at least one equilibrium that is not an FSE. This completes the proof.
\end{proof}
\vspace{-0.33cm}
\bibliographystyle{IEEEtran}
\bibliography{COSA}
\vspace{-0.6cm}
\begin{IEEEbiography}   [{\includegraphics[width=1in,height=1.25in,clip,keepaspectratio]{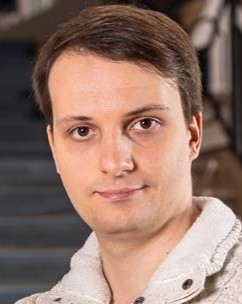}}]{Anthony Couthures} received the B.S. degree in fundamental mathematics from the University of Sciences and Technologies, Bordeaux, France, in 2017, the M.S. degree in Research Mathematics on Analysis, Modeling, Simulation in 2021, and the M.S degree in Optimization in 2022 from the University of Paris-Saclay, Orsay, France. 
    
He is currently a Ph.D. student at Université de Lorraine in CRAN, CNRS in Nancy, France. His research interests include modeling and analysis of nonlinear multi-agent systems coupled with the environment using game theory, control, and optimization.
\end{IEEEbiography}

\begin{IEEEbiography}
[{\includegraphics[width=1in,height=1.25in,clip,keepaspectratio]{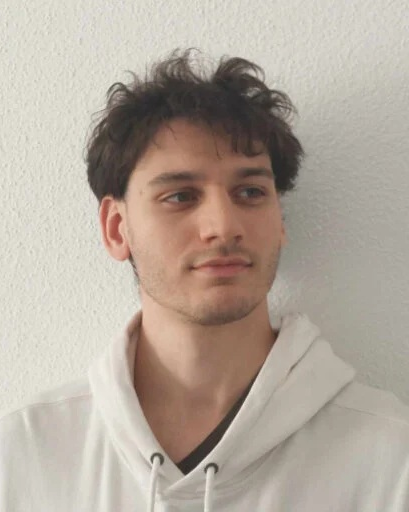}}]{Gustave Bainier} received an Engineering Master’s degree from École Centrale Nantes, France, in 2020, and a Ph.D. degree in Automatic Control from Université de Lorraine, France, in 2024. His doctoral research, carried out at CRAN, CNRS in Nancy, France, under the supervision of J.-C. Ponsart and B. Marx, focused on theoretical developments in the Linear Parameter Varying (LPV) and Takagi–Sugeno (TS) control frameworks, with applications to fault diagnosis. He is currently a Post-Doctoral Researcher with the Neuroengineering Laboratory at University of Liège, Belgium, where his research explores applications of control theory to brain-inspired computing. His research interests include LPV and TS systems, stability analysis, fault diagnosis, and neuroengineering.
\end{IEEEbiography}
\begin{IEEEbiography}    [{\includegraphics[width=1in,height=1.25in,clip,keepaspectratio]{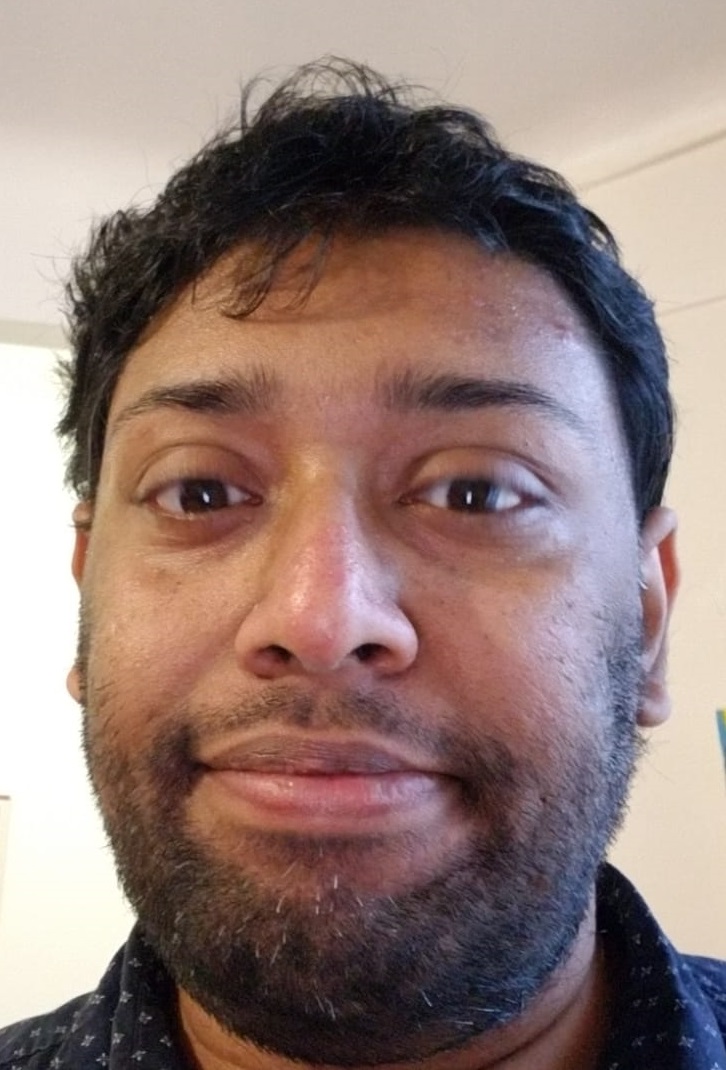}}]{Vineeth S. Varma} obtained his Bachelor’s in Physics with Honors from Chennai Mathematical Institute, India in 2008, his dual Masters in Science and Technology from Friedrich-Schiller-University of Jena in 2009 and Warsaw University of Technology in 2010. He was awarded his PhD from LSS/Paris-Saclay on energy-efficient wireless telecommunications. He did one year of post-doctoral research at Singapore University of Technology and Design from 2014-2015. Since 2016, he has been a CNRS researcher with CRAN, in Nancy, France, and got his HDR in 2024. His areas of interest are analysis, control, and games over various networks.\end{IEEEbiography}

\begin{IEEEbiography}    [{\includegraphics[width=1in,height=1.25in,clip,keepaspectratio]{{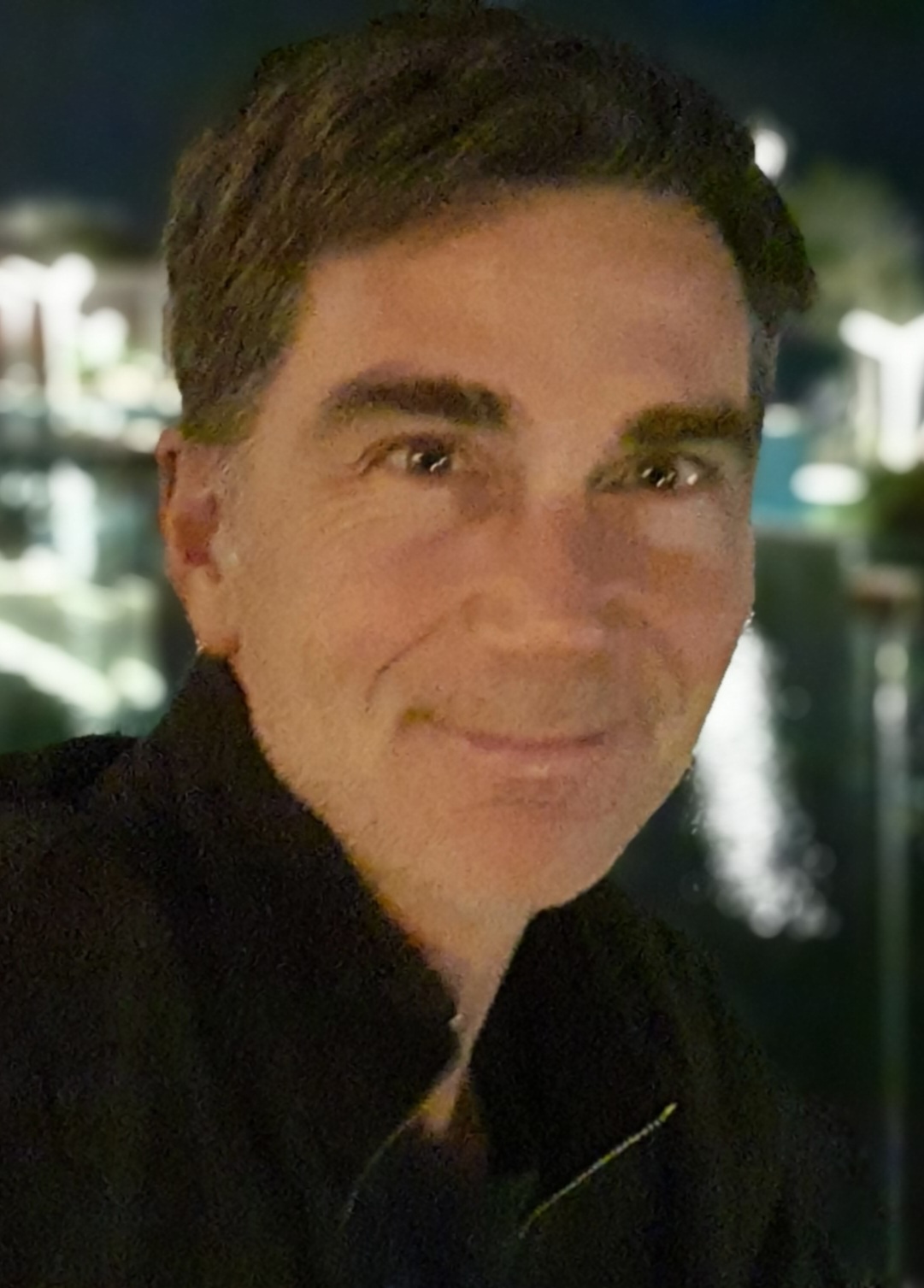}}}]{Samson Lasaulce} is a CNRS Research Director with CRAN at Nancy. From 2023 to 2025, he has been a Chief Research Scientist with Khalifa University (Abu Dhabi), where he has been the holder of the TII 6G Chair on Native AI. From 2016 to 2019, he was the holder of the RTE CentraleSupelec Chair on the "Digital Transformation of Electricity Networks". From 2002 to 2014, he was also a part-time Professor with the Department of Physics at École Polytechnique (France). Before joining CNRS he had been working for five years in private R\&D research centers (Motorola Labs and Orange Labs). His current research interests lie in distributed networks with a focus on optimization, machine learning, game theory, and optimal control. 
Dr Lasaulce has been serving as an editor for several international journals such as the IEEE Transactions. 
He is the co-author of more than 220 publications, including 15 patents and 5 books. 
Dr Lasaulce is also the recipient of several awards from the IEEE society and the Blondel Medal award from the SEE French society.\end{IEEEbiography}

\begin{IEEEbiography}    [{\includegraphics[width=1in,height=1.25in,clip,keepaspectratio]{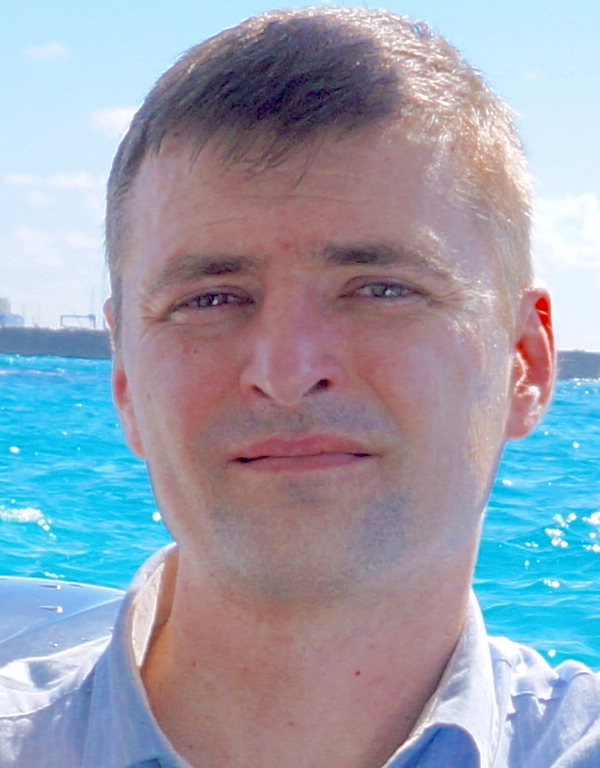}}]{Irinel-Constantin Mor\u{a}rescu} is currently Full Professor at Universit\'e de Lorraine and researcher with CRAN, CNRS in Nancy, France. He holds a Ph.D. in Mathematics and a Ph.D. in Technology of Information and Systems (received in 2006 from the University of Bucharest and the University of Technology of Compiègne, respectively), and received his "Habilitation à Diriger des Recherches" from the Université de Lorraine in 2016.  His works mainly focus on stability and tracking for different classes of hybrid systems, consensus and synchronization problems. He currently serves as a Senior Editor for IEEE Control Systems Letters and he served as Associate Editor for Nonlinear Analysis: Hybrid Systems, IEEE Control Systems Letters, and IMA Journal of Mathematical Control and Information. He was a member of the CSS Conference Editorial Board, and he is a member of the IFAC Technical Committee on Networked Systems.
\end{IEEEbiography}
\addtolength{\textheight}{-12cm}

\end{document}